%% file: main.tex
\documentclass[12pt]{article}%
\usepackage{soul}
\usepackage{subfloat}
\usepackage{subcaption}
\usepackage{amssymb,fullpage}
\usepackage{amsfonts}

\usepackage{amsmath}
\usepackage{enumerate} 
\usepackage[inline,shortlabels]{enumitem}

\usepackage{graphicx}
\usepackage[titletoc]{appendix} 
\usepackage[group-separator={,}]{siunitx}
\usepackage{hyperref}
\usepackage{cleveref}

\usepackage[usenames,dvipsnames,svgnames,table]{xcolor}
\usepackage[ruled]{algorithm2e} 

\SetAlFnt{\small}
\SetAlCapFnt{\small}
\SetAlCapNameFnt{\small}
\SetAlCapHSkip{0pt}
\IncMargin{-\parindent}
\usepackage{tikz}
\usetikzlibrary{arrows.meta}
\usetikzlibrary{arrows}

 \numberwithin{equation}{section}

\usepackage{amsfonts, amsthm, thmtools,graphicx,latexsym,url,epsf, algpseudocode, verbatim, setspace, float, paralist, mathtools}

\DeclareCaptionLabelFormat{subfig}{\figurename #1~\arabic{figure}.\Alph{subfigure}:}
\usepackage{natbib}
\def\cite{\citep}
\bibpunct{[}{]}{,}{a}{}{;}

\usepackage{hyperref}
\hypersetup{
  colorlinks,
  citecolor=blue,
  linkcolor=blue,
  urlcolor=blue}

\newcommand{\bmath}{\begin{equation}}
\newcommand{\emath}{\end{equation}}
\newcommand{\bmathnn}{\begin{eqnarray*}}
\newcommand{\emathnn}{\end{eqnarray*}}

\declaretheorem[numberwithin=section]{theorem}
\declaretheorem[sibling=theorem]{lemma}
\declaretheorem[sibling=theorem]{proposition}

\declaretheorem[sibling=theorem]{corollary}

\declaretheorem[sibling=theorem]{remark}

\declaretheorem[sibling=theorem]{definition}

\declaretheorem[sibling=theorem]{condition}

\def\setminus{-}
\def\1{{\bf{1}}}

\DeclareMathOperator{\var}{Var}

\usepackage{color}

\usepackage{nomencl}
\makenomenclature

\begin{document}

\title{Algorithms Using Local Graph Features to Predict Epidemics}

\author{ 
 \, Yeganeh Alimohammadi\thanks{Department of Management Science and Engineering, Stanford University.  \protect\url{yeganeh@stanford.edu}
  },\,
       Christian Borgs\thanks{Department of Electrical Engineering and Computer Science, U.C. Berkeley \protect\url{borgs@berkeley.edu}
   },
   \, Amin Saberi\thanks{Department of Management Science and Engineering, Stanford University.  \protect\url{saberi@stanford.edu}
  }
}

\maketitle

\begin{abstract}

We study a simple model of epidemics where an infected node transmits the infection to its neighbors  independently with probability $p$. This is also known as the independent cascade or Susceptible-Infected-Recovered (SIR) model with fixed recovery time. The size of an outbreak in this model is closely related to that of the giant connected component in ``edge percolation'', where each edge of the graph is kept independently with probability $p$, studied for a large class of networks including configuration model \cite{molloy2011critical} and preferential attachment \cite{bollobas2003,Riordan2005}. Even though these models capture the effects of degree inhomogeneity and the role of super-spreaders in the spread of an epidemic,  they only consider graphs that are locally tree like i.e. have a few or no short cycles.
Some generalizations of the configuration model were suggested to capture local communities, known as household models \cite{ball2009threshold}, or hierarchical configuration model \cite{Hofstad2015hierarchical}.

Here, we ask a different question:  what information is needed for general networks to predict the size of an outbreak? Is it possible to make predictions by accessing the distribution of small subgraphs (or motifs)?  We answer the question in the affirmative for large-set expanders with local weak limits (also known as Benjamini-Schramm limits). In particular, we show that there is an algorithm which gives a $(1-\epsilon)$ approximation of the probability and the final size of an outbreak by accessing a constant-size neighborhood of a constant number of nodes chosen uniformly at random.  
We also present corollaries of the theorem for the preferential attachment model, and study generalizations with household (or motif) structure. The latter was only known for the configuration model. 

\end{abstract}

{\footnotesize\textbf{Keywords:} Infection process, local weak convergence, contact-tracing, configuration model, preferential attachment, expanders, local features. }

\thispagestyle{empty}


\onehalfspacing


\newpage

\setcounter{page}{1}

\section{Introduction}

One of the central tasks of   epidemiology is to predict the likelihood or the size of a nascent or future epidemic from observed data. In the simplest model, the Susceptible-Infected-Recovered (SIR) model with homogeneous mixing, it is assumed that all individuals have the same probability of infecting one another, and the predictions rely on a single parameter, the basic reproduction number.  At the same time, it is known for decades, going back at least to  the study of the AIDS epidemics \cite{Gupta1989NetworksOS}, that the underlying network of contacts and their structure play an important role in the spread of the infection. 

From a mathematical  point of view, the effect of network structure on the spread of an epidemic has been analyzed for many random network models, including 
Erd\"os-Renyi graphs  \cite{neal2003sir},
configuration models \cite{DhersinMoyalChiTran,janson2014law},  and preferential attachment models {\cite{pastor2001epidemic, SIS_BBCS}}. These models offer important insights about the spread of the infection and especially the devastating effects of high degree nodes. 

{{One may be tempted to use the above network models to predict the likelihood and size of an epidemic from observed data by first  estimating the parameters of the model, and then using the resulting network to make predictions about the epidemic. 
But this {approach faces}}}
several challenges. First of all, it is not even clear which class of models to use: configuration models, preferential attachment models, inhomogeneous random graphs \cite{BollobasJansonRiordon} or variants of those allowing for community structure? And even if we could agree on a class of models, it is unclear how predictions may change under network misspecifications by the slightest amount.

In the setting of dense graphs,
the question of network misspecification and robustness can be addressed using the notions of  quasi-random graphs and graph limits.  In particular,  for dense, exchangeable networks, the Aldous-Hoover Theorem \cite{aldous1981representations,hooverrelations} determines the class of models to be considered,
{the estimation of the underlying (graphon) model can be framed as a well defined non-parametric estimation problem \cite{wolfe2013nonparametric}, }
and the theory of graph limits \cite{borgs2008convergent,borgs2012convergent} allows one to analyze deviations from these models {in terms of the so-called cut-metric.  However,  this is not a feasible approach for bounded degree networks, because there is no single natural random model for graphs with bounded average degree.

Therefore, instead of formulating the problem as a network estimation problem, 
we take a different approach: we propose an algorithm, which predicts parameters like the threshold and {the} size of an epidemic by probing and observing local neighborhoods of random nodes, and  show that those predictions  are asymptotically consistent  if (a) the network converges in the weak local sense, and (b) is sufficiently connected or more precisely, a large-set expander.  These assumptions are true for many of the random networks mentioned above. At the same time the predictions are robust in the sense that they are made without ever explicitly assuming or estimating an {underlying} random network model.

In what follows, we will describe the the key notions and definitions and set the stage for our two main theorems and algorithm. 

\vspace{5.pt}
{\em Epidemic model {and percolation}:} We consider a simple model   where an infected node transmits the infection to its neighbors independently with probability $p$. This is   known as  independent cascade \cite{Kempe-Kleinberg} {and turns out to be equivalent to} the Susceptible-Infected-Recovered (SIR) model with fixed recovery time \cite{kiss2017mathematics}. 
{In the latter model, infected nodes stay infected for} a fixed time (say one unit of time) and {while infected,  transmit the disease to neighbors} with rate
$\lambda$.  Thus,
an infected vertex has an opportunity to infect any of its neighbors independently with probability $p=\frac{\lambda}{\lambda+1}$. 

 The size of an outbreak in this model is closely related to that of the giant connected component in ``edge percolation'', where each edge of the graph is kept independently with probability $p$, studied for a large class of {random} network {models} including the configuration model \cite{molloy2011critical} and preferential attachment models \cite{bollobas2003,Riordan2005}. 
{Phrased in terms of the consequences for the spread of an epidemic}, one of the main implications of the above line of work is  that with high probability, an infection either dies out quickly (after reaching a constant number of nodes) or it spreads to a linear fraction of vertices.  Our first theorem substantially generalizes these results using the notion of weak local limit for graphs, as defined below.

\vspace{5.pt}
{\em Local weak limit:}
 Roughly speaking, a sequence of (possibly random) graphs $\{G_n\}_{n\in\mathbb N}$
 is said to have a \emph{local weak limit in probability}, if the distributions of the neighborhoods of randomly sampled vertices converge in probability.  The limit is then a probability measure $\mu$ on 
the space $\mathcal{G}_*$ of
rooted, locally finite, connected graphs ({more precisely, on equivalence classes under} graph isomorphisms which map roots into roots).
We will use
$(G,o)$ for a graph $G$ with root $o$ in $\mathcal{G}_*$,
{where $G$ should be considered unlabeled, except for the label $o$ for the root.}

{Consider now a sequence of graphs $\{G_n\}_{n\in\mathbb N}$ with graph limit $(G,o)$ distributed according to some $\mu$.}
We will study the percolated graph $G(p)$, the subgraph obtained from $G$ by independently keeping each edge with probability $p$, {and}
the component of the root in $(G(p),o)$, {to be denoted by} $C(o)$.
{As we will see,}  the probability $\zeta(p)$ of  $C(o)$ being infinite, $\zeta(p)=\mu\big(\mathbb P_{G(p)}(|C(o)|=\infty )\big)$, 
{turns out to the limit of the probability that an infection starting from a random vertex in $G_n$ leads to an infection -- under some technical conditions.  One of them will be what we refer to as}
\emph{smoothness of $\mu$ {at}  
$p$}, defined as the continuity of
 $\zeta(p)$ at $p$. 
 
The second one will be an expansion property for $G_n$.

\vspace{5.pt}
{\em Large-set edge expansion:}
 Given a graph $G=(V,E)$ and a constant $\epsilon<1/2$,  define
\begin{equation}
\label{phi}\phi(G,\epsilon)=\min_{A\subset V: \epsilon |V|\leq |A|\leq |V|/2}
\frac{e(A,V\setminus A)}{|A|}
\end{equation}
where  $e(A,V\setminus A)$ is the number of edges joining $A$ to its complement.
Call a graph $G$ an $(\alpha,\epsilon,\bar d)$ large-set (edge) expander if the average degree of $G$ is at most $\bar d$ and $\phi(G,\epsilon)\geq \alpha$.  A sequence of possibly random graphs $\{G_n\}_{n\in\mathbb{N}}$ is  called a \emph{large-set  (edge)  expander sequence with bounded average degree},  if there exists $\bar d$ 
and $\alpha>0$ such that for  any $\epsilon\in (0,1/2)$ the probability that $G_n$ is an  $(\alpha,\epsilon,\bar d)$ large-set  (edge)  expander goes to $1$ as $n\to\infty$.


\vspace{5.pt}
 \textbf{Outbreaks and their asymptotics.} Let $\mathcal{N}(G)$ be the random variable {giving}
the number of nodes that eventually get infected if one uniformly random vertex in $G$ was infected initially. {Note that there are two sources of  randomness for} $\mathcal{N}(G)$: the randomness of the infection process and the choice of {a random} seed. The next theorem shows that under certain assumptions, with high probability, either the infection dies out after infecting a constant number number of vertices or it reaches a linear fraction of vertices. 
More precisely, $\frac{\mathcal N(G)}{n}$ converges to a distribution on $[0,1]$ with one atom at zero and possibly  a second atom in $(0,1]$.



\begin{theorem}\label{thm: outbreak size} Let $\{G_n\}_{n\in\mathbb{N}}$ be a sequence of (possibly random) large-set expanders with bounded average degree that converges in probability in the local weak sense to $(G,o)\in\mathcal{G}_*$ with distribution $\mu$. {Let} $p\in[0,1]$, {and assume that}
$\mu$ {is}
smooth
{at} 
$p$. Then\\
\begin{enumerate*}
            \item\label{Thm1.1-1} The final infection size $\mathcal N$ {in an SIR model with constant recovery time} is  either  linear in $n$ or {of order}  $O(1)$. Formally, if  $\omega(n)$ is {any} function 
            such that
           $\lim_{n\rightarrow \infty}\omega(n)=\infty$
            and $\lim_{n\rightarrow \infty}\omega(n)/n=0$, then
         $ \mathbb{P} \Big(\omega(n)\leq \mathcal N\leq \frac{n}{\omega(n)}\Big)\rightarrow 0,$  where  $\mathbb P$ denotes probabilities with respect to the infection process and the randomness of $G_n$.\\
         
         \item The random variable  $\mathcal N/n$ converges {weakly in probability} to the two-atom random variable
         $\chi_p=\begin{cases}
            \zeta(p) & \text{with probability }\zeta(p),\\
            0 & \text{with probability }1-\zeta(p).
         \end{cases}$
\end{enumerate*}
\end{theorem}

A few explanations are in order. First of all, the necessary conditions of the theorem are relatively weak. In Section~\ref{sec: proofs application}, we show that the popular network models including preferential attachment and configuration models satisfy them, even after incorporating additional community (or motif) structures. In particular the above theorem characterizes the size and probability of outbreaks for configuration model \cite{janson2014law}, preferential attachment 
{(see Corollary~\ref{cor: PA}})},
configuration model with community structure \cite{trapman2007analytical,ball2010analysis, Hofstad2015hierarchical} and preferential attachment with communities {(see Corollary~\ref{thm: motif based PA})}. As far as we know, the results on preferential attachment {(with or  without communities)} was not known in the literature.

More importantly, the graph sequences considered in the theorem do not have to be `perfectly' drawn from a random model for the result to hold. In fact, given an arbitrary random graph model with smooth limit (a property which holds for most random graph models, except possibly at the threshold), they can be `quasi-random' in the sense that the they are (random or) deterministic graphs with the same local limit.
As long as the sequence  is a large-set expander, the size and probability of an infection on the sequence will be the same as for the random model.

Finally, it is worth noting that some notion of connectivity is necessary for the above theorem to hold. This can be seen by comparing $d$-regular random graphs on $n$ nodes to the union of two disjoint $d$-regular random graphs on $n/2$ nodes that are connected with one edge
as described in {Remark~\ref{remark:  lwc doesnt imply locality}}. 


Next, we will state the algorithmic implications.

\vspace{5.pt}
\textbf{Algorithm for estimating the probability and size of outbreak.}
Suppose that we do not have access to the underlying contact network $G$. Instead, we have access to a black-box algorithm that takes a constant $k=O(1)$ as an input and
runs the SIR process starting from a uniformly random node. Then the algorithm returns whether  $k$ other nodes get infected eventually. 
 So, for each query we get a $0-1$ answer,  $1$ {when} starting
 {an} infection from the chosen vertex has reached  $k$ other nodes,
 {and $0$ otherwise} (see Algorithm \ref{alg: forward process}). 
 

 \begin{algorithm}[h]
 \caption{Local infection algorithm}
 \label{alg: forward process}
 A constant $k>0$ and the contact network $G$.

Draw a uniformly random node $v$.

 Find $B_k(v)$, {the set of} the nodes of distance at most $k$ from $v$.

 Run SIR process in $B_k(v)$ starting from $v$:

\Indp Initialize susceptible nodes $S=B_k(v)\setminus \{v\}$, infected nodes $I=\{v\}$, and removed nodes $R=\emptyset$.

 {\While{$I\neq\emptyset$}{
Let $u$ be the next node in $I$.

Remove infected neighbors of $u$ in $S$ and add them to $I$.

Remove $u$ from $I$ and add it to $R$.
 }}

\Indm
\eIf{$|R|\geq k$}{Return $1$.}{Return $0$.}
\end{algorithm}
  
  Let $\widetilde{\mathcal N}(k,q)$ be the average of answers for $q$ independent queries to Algorithm~\ref{alg: forward process} with $k$ as an input.
We prove that for any $\epsilon>0$, there exist constants $k_{\epsilon}$ and $q_{\epsilon}$ independent {of} $n$ such that $\widetilde{\mathcal N}(k_\epsilon,q_\epsilon)$ gives an $(1-\epsilon)$-approximation of both the probability of an outbreak and 
the relative size of an outbreak.

\begin{theorem}\label{thm: local alg}
Let $\{G_n\}_{n\in\mathbb{N}}$ and $\zeta(p)$
satisfy the conditions in Theorem \ref{thm: outbreak size}. 
Then for any given $\epsilon>0$ there exist integers $k_{\epsilon}, q_\epsilon$,  such that the average output of $q_\epsilon$ independent runs of Algorithm~\ref{alg: forward process} with $k_\epsilon$ as an input is a $(1-\epsilon)$-approximation of the probability and the size of an outbreak
{in an SIR model with constant recovery time}. Formally, there exists an $n_\epsilon>0$,  such that for all $n\geq n_\epsilon$
    \[\mathbb{P}\big(|\widetilde{\mathcal N}(k_\epsilon,q_\epsilon)-\zeta(p)|\geq \epsilon\big)\leq \epsilon. \]
Here $\mathbb P$ denotes the probabilities first over the randomness of $G_n$ and then the infection process and the sampled vertices contributing to $\widetilde{\mathcal N}(k_\epsilon,q_\epsilon)$.
\end{theorem}


Note that while Theorem \ref{thm: local alg} requires that the graph sequence converges, we do not need to use the properties of or even know the limit graph. We can determine the size and the probability of an epidemics only by sampling the local neighborhoods of a constant number of nodes.

In the next section, Section \ref{sec: intro random graph}, we show that our main theorems are applicable to many known random graph models. In addition, in Section \ref{sec: intro household model}, we introduce a {new class of}
random graph model with community {structure},
generalizing several older models, and apply our result to the new set of models. 

\subsection{Application: Classical Random Graph Models}\label{sec: intro random graph} 

In the following section, we consider examples of graphs for which  we can apply the results of Theorems~\ref{thm: outbreak size} and \ref{thm: local alg}. 
We show the following models converge in a local weak sense and $\zeta(p)$ on their limit is continuous for every $p\in[0,1]$.
\begin{enumerate*}
    \item The preferential attachment model.
    \item The configuration model.
\end{enumerate*}

\subsubsection{Preferential Attachment Models}\label{sec: intro PA}
Preferential attachment models describe a dynamic population where each new individual contacts the older individuals proportional to their degree. More precisely, the model has a parameter $m\in \mathbb N$, and is defined as follows.  Starting from a connected  graph  
$G_{t_0}$
on at least $m$ vertices, a random graph $G_t$ is defined inductively: given $G_{t-1}$ and its degree sequence $d_i(t-1)$, we form a 
new graph by adding one more vertex, $v_t$, and connect it
to $m$ distinct  vertices $w_1,\dots,w_m\in V(G_{t-1})$ by first choosing $w_1,\dots,w_m\in V(G_{t-1})$ {i.i.d with distribution $\mathbb P(w_s=i)=\frac {d_i(t-1)}{2|E(G_{t-1})|}$, $s=1,\dots m$,} and then conditioning on all vertices being distinct (thus avoiding  multiple edges).    We denote the resulting random graph sequence by $({P{\hskip-.2em}A}_{m,n})_{n\geq 2m+1}$, and following \cite {berger2014}, we call the version of preferential attachment we defined above the \emph{conditional model}, while the model where the conditioning step is left off will be called the independent model.

The next lemma, which is a consequence of previous work on the preferential attachment graphs,  will allow us to apply Theorems \ref{thm: outbreak size}-\ref{thm: local alg} to $PA_{m,n}$.
 
\begin{lemma}\label{lm: PA- expansion + continuity}
Let $m\geq 2$, and 
for a positive integer $n\geq 2m+1$ let ${P{\hskip-.2em}A}_{m,n}$ be the conditional preferential attachment graph defined above.  Then the following holds.
\begin{enumerate*}
\item \label{part: continuity PA} There exists a distribution $\mu$  on $\mathcal G_*$ such that the sequence $\{{P{\hskip-.2em}A}_{m,n}\}_{n\geq 2m+1}$ converges in {the} local weak sense to $(G,o)$ distributed as $\mu$. Also, for all $p\in[0,1]$ the percolation function $\zeta(p)=\mu\Big(\mathbb{P}_{G(p)}C(o)=\infty\Big)$ is continuous.\\
    \item \label{part: expansion PA} Then there exists some $\alpha>0$  independent from $n$ such that {for all $\epsilon\in (0,1/2)$, with probability tending to $1$,} $\{{P{\hskip-.2em}A}_{m,n}\}_{n\geq 2m+1}$ is an $(\alpha,{\epsilon},2m)$
       large-set expander.
\end{enumerate*}
\end{lemma}
The local weak convergence of preferential attachment models is well-established in \cite{berger2014}, where they showed the limit is a P\'olya-point process.  {Moreover, the continuity of the survival probability and the positive expansion of $PA_{m,n}$ was shown in \cite{abs2021locality}. 
}
 As a consequence of the above lemma the size and the probability of an outbreak is local  under the preferential attachment model.
\begin{corollary}\label{cor: PA}
Let $m\geq 2$,  and let $\{G_n\}_{n\geq 2m+1}=\{{P{\hskip-.2em}A}_{m,n}\}_{n\geq 2m+1}$ be the conditional preferential attachment graph defined above.  Then the conclusions of Theorems \ref{thm: outbreak size} and \ref{thm: local alg} hold for all $p\in[0,1]$. 
\end{corollary}
\begin{proof}
As a direct consequence of Lemma~\ref{lm: PA- expansion + continuity} we can apply Theorems \ref{thm: outbreak size} and \ref{thm: local alg}.
\end{proof}

\subsubsection{Configuration Model}
The common way to construct a random graph given a degree sequence is the configuration model, introduced by \cite{bollobas2001random}.
Let $\mathbf d_n= (d_1,\ldots,d_n)$ be a graphical degree sequence , i.e., there exists at least one simple graph with this degree sequence. Then to construct the configuration model $CM(\mathbf{d}_n)$, consider $d_i$ half-edges for the vertex $i$.  The half-edge are then matched uniformly,  in general resulting in a multi-graph with self-loops and multiple edges.

The uniform configuration model is defined as the random graph model whose samples are uniformly random among all \emph{simple} graphs with degree sequence $\mathbf d_n=(d_1,\ldots,d_n)$.  We denote such a sample by $CM^*(\mathbf d_n)$.  To generate a sample from the uniform model, one often first generates a sample from the so-called configuration model.  It is well known that if the degree sequence is graphical,  conditioning $CM(\mathbf d_n)$ to be a simple graph, gives  a  sample  from the uniform model $CM^*(\mathbf d_n)$. See Chapter 7.5 in \cite{hofstadVol1} for the definition and properties of the uniform  model and its relationship to the configuration model.

  In order to apply our results to the configuration model, we first need to describe the conditions for the existence of the local weak limits. 
  To formulate it, let $D_n$ be {the} degree of a uniform random node in the degree sequence $\mathbf d_n$.
 Recall the definition of the empirical distribution $F_n$ of the sequence $\mathbf d_n$,
\[F_n(x)=\frac{1}{n}\sum_{j\in[n]}\mathbf1_{d_j\leq x}. \]
\begin{condition}
\label{cond: CM regular degree}
Let $\mathbf d_n=(d_1,\ldots,d_n)$ be a graphical degree sequence. Then there exists a random variable  $D$ such that\\
\begin{enumerate*}[label=(\roman*)]
    \item \label{cm1} $\mathbb P(D\geq 3)=1$., and as $n\to\infty$,
    $D_n\overset{d}{\to} D, \qquad \text{  and}\qquad \mathbb{E}[D_n]\rightarrow \mathbb{E}[D]<\infty.$
    \\
    \item \label{cm2} The empirical distribution $F_n$ satisfies,
    $[1-F_n](x)\leq c_Fx^{-\tau-1},$
    for some $c_F>0$ and $\tau\in(2,3)$.
\end{enumerate*}
\end{condition}

\begin{theorem}\label{thm: CM}
{Let $\{G_n\}_{n\in\mathbb N}=\{CM^*(\mathbf d_n)\}_{n\in\mathbb N}$, where $\{\mathbf d_n\}_{n\in\mathbb N}$ is a}
graphical degree sequence  satisfying  Condition~\ref{cond: CM regular degree}.
Then the conclusions of Theorems \ref{thm: outbreak size} and \ref{thm: local alg} hold  for all $p\in[0,1]$.
\end{theorem} 
\begin{proof}
The expansion of $CM^*(d_n)$ was shown in Lemma 12 in \cite{abdullah2012cover}  under the Condition~\ref{cond: CM regular degree} \ref{cm1}.
Also,  the local weak convergence of $CM^*(d_n)$ to a uni-modular branching processes was established under the Condition~\ref{cond: CM regular degree} \ref{cm1}-\ref{cm2} in Theorem 4.5 in \cite{RemcoVol2}. Since the limit is a branching process the continuity of survival probability is already known (see e.g., \cite{broman2008survival} for a more general result). Hence, all the conditions needed for Theorems \ref{thm: outbreak size} and \ref{thm: local alg} hold for $\{CM^*(\mathbf d_n)\}_{n\in\mathbb N}$.
\end{proof}

\subsection{Application: Motif-Based Graph Models}\label{sec: intro household model}
{
One of the important characteristics of   real-world networks is that they contain many short cycles \cite{scott1988social}. However, in most random graph models including the configuration model and preferential attachment model, the opposite is true: the neighborhood of a typical node looks like a tree. 

One approach for addressing this problem is to incorporate the community structure directly, for example by replacing the nodes with dense graphs.  For example. Trapman~ \cite{trapman2007analytical}  replaces each vertex of degree
 $k$ in $G_n$  {by} a complete graph of size $k$ with some fixed probability. The complete graphs represent households, where each individual interact with all others in their household, and {has} exactly one contact outside of the household. A  more general model by Ball et. al. \cite{ball2009threshold} extends the configuration model by allowing for drawing households independently from a family of graphs. 
 
We extend this framework.} Our model has two ingredients:
a network of external connections (which at the current level of generality can be arbitrary), and a distribution on small graphs (motifs), describing internal connections of a group of people (see Fig.~\ref{fig:motif}).

\begin{figure}[!hbt]
    \centering
    \input{household_fig}
    \caption{A motif-based graph with 1) an internal level representing connections (solid edges) between individuals in one motif (colored areas), and 2) an external level representing  connections (dashed edges) between motifs .}
    \label{fig:motif}
\end{figure}
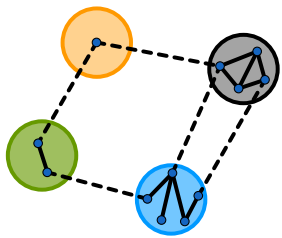

{Let us first define the notion of a motif, which is a pair} $M=(F, (d_v^{ext})_{v\in V(F)})$ contains (i) a simple graph $F$, giving the internal connection of the nodes in the motif, and (ii) {a collection of external degrees,} $(d_v^{ext})_{v\in V(F)}$, {one for each} vertex in $F$.  For a motif $M$, let $d(M)=\sum_{v\in F} d_v^{ext}$ be the total external degree of the motif $M$, and $v(M)=|F|$ be number of vertices in the motif.  

Now, to define a motif-based graph,  specify a collection of probability distributions $\mathcal M_d$ over motifs  $M$ with $d(M)=d$, one for each $d$. Then, given such a collection, plus an arbitrary external graph $G^{ext}$ on $n$ vertices,  replace all vertices of degree $d$ in $G^{ext}$ independently at random by a motif drawn from $\mathcal M_d$, {resulting in a larger graph $G$ containing precisely $n$ motifs.}  

\begin{condition}
\begin{enumerate*}[label=(\roman*)]
\item \label{cond: internal} (Internal regularity) Let $\mathcal M$ be the space of all motifs. Then there exists a constant $S_{\max}<\infty$ such that $\sup_{M\in\mathcal M}v(M)\leq S_{\max}$.\\

\item \label{cond: external}(External regularity)
There exists a probability distribution $\mu^{ext}$ on $\mathcal G_*$ such that the sequence of (possibly random) external graphs $\{G_n^{ext}\}_{n\in\mathbb N}$ with a bouned average degree converges in the local weak sense to $(G^{ext},o)\in\mathcal G_*$ with distribution $\mu^{ext}$. 
\end{enumerate*}
\end{condition}

The combination of 
internal and external regularities is sufficient  for local weak convergence. 

\begin{lemma}\label{lm: LWC of MRG}
Let $\{G_n\}_{n\in\mathbb N}$ be a sequence of motif-based graphs satisfying the Condition \ref{cond: internal}
and \ref{cond: external}. Then there exists a probability distribution $\mu$ on $\mathcal G_*$ such that   $\{G_n\}_{n\in\mathbb N}$ converges in the local weak sense to $(G,o)$ drawn from distribution $\mu$. 
\end{lemma}

{
\begin{corollary}\label{thm: motif based general}
Let $\{G_n\}_{n\in\mathbb N}$ be a sequence of motif-based graphs satisfying the Condition \ref{cond: internal} and \ref{cond: external}.  If the external graph is a large set expander, then  the conclusions of Theorems \ref{thm: outbreak size} and \ref{thm: local alg} hold for all $p\in[0,1]$ for which $\zeta(p)$ is continuous.
\end{corollary}
}

On the way of proving 
{this corollary}, 
we prove a general result that might be of independent interest:  large-set edge expansion along with the existence of local weak limit implies large-set vertex expansion (see Lemma~\ref{lm: vertex expansion}).
Finally, we consider special cases of the motif-based model where the external graph has the same limit as one of the models discussed in Section \ref{sec: intro random graph}.

\begin{corollary}[Motif-Based Preferential Attachment]\label{thm: motif based PA}
Let $m\geq 2$, and let $\{G_n\}_{n\geq 2m+1}$ be a sequence of motif-based graphs such that
$G_n^{ext}={P{\hskip-.2em}A}_{m,n}$.
If the motif-distribution satisfies
Condition \ref{cond: internal},
then the conclusions of Theorems \ref{thm: outbreak size} and \ref{thm: local alg} hold for all $p\in[0,1]$.
\end{corollary}

\begin{corollary}[Motif-Based Configuration Model]\label{thm: motif based CM}
Let $\{G_n\}_{n\in\mathbb N}$ be a sequence of motif-based graphs such that $G_n^{ext}=CM^*(\mathbf d_n)$ for a graphical degree sequence $\{\mathbf d_n\}_{n\in\mathbb N}$ satisfying  Condition~\ref{cond: CM regular degree}.
If the motif distribution
satisfies Condition \ref{cond: internal}, then
 the conclusions of Theorems \ref{thm: outbreak size} and \ref{thm: local alg} hold  for all $p\in[0,1]$.
\end{corollary}

 Percolation on general household model with external configuration model was studied by Hofstad et. al.\cite{Hofstad2015hierarchical}. In particular, they proved that if the expected size of household size and the average degree degree is finite, then the size of the giant after percolation converges in probability. They  also compared the spread of an epidemics on motif-based configuration model with real-world networks empirically \cite{stegehuis2016epidemic}.  
 {In Section~\ref{sec: proof infection}, we will give a different proof for the result in \cite{Hofstad2015hierarchical} by showing the large-set expansion and the smoothness of the limit motif-based configuration models. Our proof also applies to motif-based preferential attachment model along with other motif-based models, as stated in Corollary~\ref{thm: motif based general}. }

\subsection{Other Related Work}
\paragraph{Epidemics on networks.}
The study of epidemics on networks has got much attention in different communities. 
The work of \cite{newman2001random} led to heuristics differential equations describing the evolution of SIR process on configuration models \cite{Volz2008SIRDI,VolzMeyers}. Later works by  \cite{decreusefond2012large,janson2014law} made the heuristics rigorous for the configuration model. In some special cases, the SIR process has been studied on trees \cite{sharkey2015exact} and grids \cite{n2019sir}.  
See the book \cite{kiss2017mathematics} for the results in theoretical biology and physics. The discrete version of the Markovian SIR process that we consider,  is also referred to as the Reed-Frost model (see e.g., \cite{picard1990unified}).

Also related to our work is the result of Ball et. al. \cite{ball2009threshold} who study epidemics on networks in which the motifs are complete graphs and derive the \emph{expected} size of the outbreak under the general SIR model \cite{ball2010analysis}.

\vspace{5.pt}
{\em Percolation.}
The study of percolation on expanders has recently received much attention in the probability community. In \cite{alon2004}, it was shown that on bounded degree expanders, there exists at most one linear size component (giant).  If in addition, one assumes a weak local limit, one obtains locality of the  threshold for the appearance of a giant \cite{benjamini2009critical,sarkar2018note}, in the sense that it  can be inferred from the limit. This result was generalized to Cayley graphs \cite{martineau2017locality}.

The relative \emph{size} of the giant, {which is an important technical ingredients {has only been studied recently. Earlier work studied}
the size on random trees \cite{bertoin2015supercritical}, and hypercubes \cite{van2017hypercube}, 
{as well as} the 
locality of the size of the giant on expanders for the case of bounded regular degree  with high girth 
\cite{Krivelevich20HighGirthExapnder}.  In this case, the relative size of the giant is given by the survival probability of a percolated branching process. 
{Here, our main input is the work of }{\cite{abs2021locality}, where they proved convergence of the relative size of the giant for expanders with local weak limits, and applied their result to preferential attachment models.}}

\vspace{5.pt}
{\em Network Modeling and Estimation}
The question of network misspecification and robustness has been studied in dense graphs using the notion of quasi-random graphs and graph limits.  For dense, exchangeable networks, the Aldous-Hoover Theorem \cite{aldous1981representations,hooverrelations} determines the class of models to be considered: it is the class of dense, inhomogeneous random graphs or graphon models, see \cite{orbanz2014bayesian}
 for a modern treatment of this relationship.  Misspecifications, or closeness to a random model can then be formulated as closeness in the cut-metric, with (not necessarily random) models which asymptotically ``look'' like an inhomgeneous random graph being precisely those which converge to the same graphon \cite{borgs2008convergent,borgs2012convergent}. 
In this context, estimation  becomes a non-parametric estimation problem which has been well studied, see  \cite{wolfe2013nonparametric} 
where this problem was first formally formulated as a graphon estimation problem (note, however, that the basic idea is much older, and goes back to at least \cite{bickel2009nonparametric}).
There are generalizations of this approach for sparse networks with divergent average degree, see \cite{borgs2017graphons} for an overview, but there is no one natural random model for graphs with bounded average degree.

\section{Proof of the Locality of Infection Processes}\label{sec: proof infection}
The first observation needed for proving Theorem~\ref{thm: local alg} is  that infection processes with fixed recovery times can be coupled to unoriented percolation.  Section~\ref{sec: coupling} is on this coupling. It will prove that when the relative size of the giant after percolation is concentrated, then the infection is also  ``local'' in the sense that it can be read off from its local weak limit. In Section~\ref{sec: giant}, we shift our focus to percolation, {and the use of some recent results to show that first the size of the giant in large-set expanders converges to its limit, and  second the size of the connected component of a random node can be bounded by a constant with high probability.}  Building upon the coupling and the results on percolation, we prove Theorems~\ref{thm: outbreak size} in Section~\ref{sec: outbreak is two atom} and Theorem~ \ref{thm: local alg} in Section~\ref{sec: local alg}.

\subsection{Coupling  Infection Processes to Percolation}\label{sec: coupling}

Suppose more than one node can be infected initially and let $v_1,\ldots, v_k$ be the set of initially infected nodes. Also, let $\mathcal N(v_i)$ be the set of nodes getting the infection eventually from node $v_i$. The next lemma couples $\mathcal N(v_i)$ to the size of the connected component containing $v_i$ after percolation in $G(p)$. To state the result, we  use $\mathbb{P}_{G(p)}$ and $\mathbb{E}_{G(p)}$ to denote probabilities and expectations with respect to percolation on a graph $G$
\begin{lemma}\label{lm: coupling infection}
Given any graph $G$ of size $n$ and  $k\leq n$, let $v_1,\ldots, v_k$ be the  initial set of infected nodes in an SIR process with parameter $p$. Then for any $n'<n$,
\[\mathbb P\Big(\sum_{i=1}^k\mathcal N(v_i)=n'\Big)=\mathbb P_{G(p)}\Big(\mid \cup C(v_i)\mid=n'\Big).\]
\end{lemma}

\begin{proof}
 To express an instance of SIR process, for each edge $(u,w)\in E(G)$ draw $t_{u,w}$ from $exp(\frac{p}{1-p})$.  The random variable $t_{u,w}$ can be viewed as the contact time of $u$ and $w$ after either of $u$ or $w$ gets infected first.  If $t_{u,w}>1$, then the contact time happens after the recovery of the endpoint, so in this case, let $t_{u,w}=\infty$. 
 Moreover, for the convention, if an edge $(u,w)$ does not exist let $t_{u,w}=\infty$. Note that each edge in $G$ is set to $\infty$ independently with probability 
 $1-p$. So, the set of edges with a finite value can be coupled to the set of edges that are retained in $G(p)$. So, it remains to show   random variables $t_{u,v}$ specify a unique SIR process.

 To describe the SIR process using the random variables $t_{u,v}$, let $t_v$ be the time that the node $v$ gets infected. Hence, for the initial infected nodes $t_{v_1}=\cdots = t_{v_k}=0$. Start from time $t=0$. For any time $t$, let $I_t$ be the set of infected nodes that their recovery times have not passed yet. Recall that the recovery time of the node $v$ is  $t_v+1$ by the definition of the SIR process. Also, let $N_S(I_t)$ be the set of susceptible neighbors of $I_t$  (the neighbors that are not  infected yet). 
For each node $u\in N_S(I_t)$, let  
$t_u=\min_{v\in I_t} (t_v+t_{v,u})$, where the minimum is over all nodes f $u$ in $I_t$. Since all the non-neighbors have the contact time of $\infty$, they will affect $t_u$.

 Now, by the choice of contact times, if  $t_u$ is finite, then it is infected before the time $t_u$.
 Let $u\in N_S(I_t)$ be the node with minimum infection time, i.e., $u=\text{arg}\min_{u\in N_S(I_t)} t_u$. If $t_u<\infty$,  add $u$  to the set of infected nodes $I_{t'}$, for all $t_u\leq t'< t_u+1$. In addition, increase the time $t$ to $t_u$ and repeat the above process until we get to a time that the set of infected nodes are empty. 
 
 In the above procedure, we use each $t_{u,v}$ only for updating from at most one side. This is because either $v$ gets infected before $u$ or $u$ before $v$. So, a node is infected in the SIR process if and only if it is reachable by one of $v_1, v_2, \cdots v_k$ in the $G(p)$.  
\end{proof}

\subsection{Size of the Giant in Percolation}\label{sec: giant}
With the coupling of the infection process and the percolation in hand, the next step is the concentration of the giant component.
It was shown in in \cite{abs2021locality} that the relative size of the largest component for percolation on large-set expanders satisfying conditions in Theorem~\ref{thm: outbreak size} converges in probability to $\zeta(p)$. 

This is not enough to prove our main theorems. In fact, we need to bound the size of the connected component of a random node which is not in the largest component. 
  In fact, concentration of the size of the giant is enough to prove the locality of infection processes. For that purpose, we give the following definition.
 \begin{definition}[Graphs with converging giant]\label{defi: converging giant}
Let $p\in [0,1]$, and let $\mu$ be a probability distribution on $\mathcal{G}_*$. A  sequence $\{G_n\}_{n\in\mathbb N}$ of (possibly random) graphs with the local weak limit $(G,o)\sim\mu$ is called a sequence  of graphs with a \emph{converging giant} component  if 
\[\frac{|C_1|}{n}\overset{\mathbb P}{\to}\zeta(p),\]
where $\zeta(p)=\mu(\mathbb P_{G(p)}(|C(o)|=\infty))$.
\end{definition}
By a result of
van der Hofstad \cite{Hofstad2021AlmostLocal},
it is easy to see that if a random node in graphs with converging giant is outside the largest component then its component size is $O_{k,\mathbb P}(1)$.  
\begin{proposition}\label{prop: converging giant}
Let $\{G_n\}_{n\in\mathbb N}$ be a sequence of graphs with a converging giant. Then  for a uniform random node $v$, if $\zeta(p)>0$,
\[\lim_{k\to \infty}\lim_{n\to\infty}\mathbb P\Big( v\not\in C_1,\text{ and }|C(v)|\geq k\Big)=0,\]
and if $\zeta(p)=0$, then $C(v)=O_{k,\mathbb{P}}(1)$.
Further, for two uniform random nodes $u$ and $v$, 
\begin{equation}\label{eq: O_k undirected}
 \lim_{k\to \infty}\lim_{n\to\infty}   \frac{1}{n^2}\mathbb E(|C(v)|, |C(u)|\geq k, C(v)\neq C(u))=0.
\end{equation}
\end{proposition}
\begin{proof}[Proof of Proposition~\ref{prop: converging giant}]
By the sufficiency condition in Theorem 2.2 of \cite{Hofstad2021AlmostLocal}, for any graph sequence with a converging giant we know,
\[ \lim_{k\to \infty}\lim_{n\to\infty}   \frac{1}{n^2}\mathbb E(|C(v)|, |C(u)|\geq k, C(v)\neq C(u))=0.\]
Then it remains to show $C(v)=O_{k,\mathbb{P}}(1)$ for $v\not\in C_1$. The case that $\zeta(p)=0$ directly follows from \eqref{eq: O_k undirected}.
In the case that $\zeta(p)>0$, choose $0<\delta<\zeta(p)$. Then by the condition of converging giant for large enough $n$ with high probability $|C_1|\geq \delta n$. Therefore, with probability larger than $\delta$ a random node lies in $C_1$ and
\begin{align*}
    \lim_{k\to \infty}\lim_{n\to\infty}   \frac{1}{n^2}\mathbb E(|C(v)|, |C(u)|\geq k, C(v)\neq C(u))\geq   \lim_{k\to \infty}\lim_{n\to\infty}   \delta\mathbb P( |C(v)|\geq k, v\not\in C_1),
\end{align*}
where in the right hand side the probability is over a uniform random node, and the left hand side goes to zero by \eqref{eq: O_k undirected}.
\end{proof}

{Combining this with Theorem 1.1 in \cite{abs2021locality} we immediately get the following corollary.}
\begin{corollary}\label{thm: size of giant in expander}
Let $\mu$, $\{G_n\}_{n\in\mathbb{N}}$, and $\zeta$  obey the assumptions of in Theorem \ref{thm: outbreak size}.  Let  $C_{i}$ be the $i^{th}$ largest component of $G_n(p)$.
If $\zeta$ is smooth with respect to $p$, then
\[\frac{|C_1|}{n}\overset{\mathbb{P}}{\to}\zeta(p),\]
with $\overset{\mathbb{P}}{\to}$ denoting convergence in probability with respect to both $\mu$ and percolation. Further, if $v$ is chosen uniformly at random then 
\[\lim_{k\to \infty}\lim_{n\to\infty}\mathbb P\Big( v\not\in C_1,\text{ and }|C(v)|\geq k\Big)=0,\]
and if $\zeta(p)=0$, then $|C(v)|=O_{k,\mathbb{P}}(1).$
\end{corollary}

In the following sections, we will lift the assumptions of the large-set expansion and the smoothness of $\mu$  in  Theorems~\ref{thm: outbreak size} and \ref{thm: local alg}  by assuming that the graph sequence have a converging giant. 

\subsection{Size of the Outbreak is Two-Atom}\label{sec: outbreak is two atom}

We will show that Theorem~\ref{thm: outbreak size} holds when the sequence of graphs have a converging giant. 

\begin{lemma}[Generalization of Theorem \ref{thm: outbreak size}]\label{lm: outbreak size}
Let $\{G_n\}_{n\in\mathbb N}$ be a sequence with a converging giant with the limit distributed as $\mu$ on $\mathcal G_*$. 
Then\\
\begin{enumerate*}
            \item\label{Thm1.1-1} If  $\omega(n)$ is such that
           $\lim_{n\rightarrow \infty}\omega(n)=\infty$
            and $\lim_{n\rightarrow \infty}\omega(n)/n=0$, then
         $\mathbb{P} \Big(\omega(n)\leq \mathcal N\leq \frac{n}{\omega(n)}\Big)\rightarrow 0,$
         where  $\mathbb{P}$ denotes probabilities with respect to the infection process and the randomness of $G_n$.\\
         \item\label{Thm1.1-2} Let $\chi_p$ be the  random variable that is equal to $\zeta(p)$ with probability $\zeta(p)$ and equal to $0$ with probability $1-\zeta(p)$.  Then $\mathcal N/n$ converges weakly in probability to $\chi_p$.
\end{enumerate*}
\end{lemma}

\begin{proof}
Recall that by Lemma~\ref{lm: coupling infection}, $\mathcal N$ is equal to the size of the connected component of a uniform random node,
\[\mathbb P(\frac{n}{\omega(n)}\geq\mathcal N\geq \omega(n))=\mathbb P_{G_n(p),\mathcal P_n}(\frac{n}{\omega(n)}\geq |C(v)|\geq \omega(n)),\]
where the second probability is both over percolation and choosing a random node.

First, assume $\zeta(p)=0$.
Then by Proposition~\ref{prop: converging giant},
\begin{equation}\label{eq: C(v) in limit}
    \lim_{k\to\infty}\lim_{n\to\infty}\mathbb P_{G(n)}( |C(v)|\geq k)=0.
\end{equation}
Then for any increasing function $\omega(n)$ with $\lim_{n\to\infty}\omega(n)=\infty$, 
$\mathbb P ( |C(v)|\geq \omega(n))\to 0.$
To see this, fix some $\epsilon>0$ and $n$, and let $k=\omega(n)$. Then by \eqref{eq: C(v) in limit}, there exists an $n_0>0$ such that for all  $n'\geq \max(n_0,n)$, we have $\mathbb P_{G(n')}(|C(v)|\geq \omega(n))\leq \epsilon.$ Then since $\omega$ is increasing, for $n'\geq \max(N,n)$ 
\begin{equation*}\label{eq: subcritical up bnd}
    \mathbb P_{G(n')} \big(|C(v)|\geq \omega(n')\big)\leq\mathbb P_{G(n')} \big( |C(v)|\geq \omega(n)\big)\leq \epsilon,
\end{equation*}
which proves the claim in the case that $\zeta(p)=0$. Now, since the graph sequence have a converging giant if $\zeta(p)>0$, then there exists $\epsilon$ such that for large enough $n$, $|C_1|\geq\epsilon n$. Therefore, if $v\in C_1$ then for large enough $n$, $|C(v)|\geq \frac{n}{\omega(n)}$. Then again by Proposition~\ref{prop: converging giant} and similar to above argument it is easy to see that
\[\mathbb P_{G_n(p),\mathcal P_n}(v\not\in C_1\text{ and } |C(v)|\geq \omega(n))\to 0.\]
This proves part \ref{Thm1.1-1}. 

So, it remains to prove part \ref{Thm1.1-2}. Again by Lemma~\ref{lm: coupling infection},
\[\lim_{k\to\infty}\lim_{n\to\infty}\mathbb P(\mathcal N\geq k)=\lim_{k\to\infty}\lim_{n\to\infty}\mathbb P_{G_n(p),\mathcal P_n}(|C(v)|\geq k).\]
Then consider two cases. First, assume $\zeta(p)=0$. Fix some $\epsilon>0$. Proposition~\ref{prop: converging giant}, implies that there exists $K_0$ and $N_0$ such that for all $k\geq K_0$ and $n\geq N_0$
\[\mathbb P(\mathcal N\geq k)\leq\epsilon.\]
Given $\epsilon'>0$, let $n\geq \max(N_0,k\epsilon'^{-1})$, then we have 
\[\mathbb P(\frac{\mathcal N}{n}\geq \epsilon')\leq \mathbb P(\frac{\mathcal N}{n}\geq \frac kn)\leq\epsilon,\]
which proves Part \ref{Thm1.1-2} in the first case.

Now, assume $\zeta(p)>0$.  For a fixed $\epsilon\in(0,\zeta(p)/2)$, by Lemma~\ref{lm: coupling infection},
\begin{equation}\label{eq: coupling thm1.1-2}
    \mathbb P\Big(\frac{\mathcal N}{n}\geq \epsilon\Big)=\mathbb P_{G_n(p),\mathcal P_n}\Big(\frac{|C(v)|}{n}\geq \epsilon\Big).
\end{equation}
Moreover, by  Proposition~\ref{prop: converging giant}, if a random node has a linear sized component then it is in the largest component with high probability,
i.e., for any $\epsilon'>0$,
\begin{equation}\label{eq: large comp is C1}
    \mathbb P_{G_n(p)}\Big(|\mathbb P_{\mathcal P_n}\big(\frac{|C(v)|}{n}\geq \epsilon\big)-\mathbb P_{\mathcal P_n}\big(v\in C_1\text{, and }\frac{|C(v)|}{n}\geq \epsilon\big)|\geq \epsilon'\Big)\to 0.
\end{equation}
It is obvious that, 
\begin{equation}\label{eq: obvious}
    \mathbb P_{G_n(p),\mathcal P_n}\Big(v\in C_1\text{, and }\frac{|C(v)|}{n}\geq \epsilon\Big)=\mathbb P_{G_n(p),\mathcal P_n}\Big(\frac{|C_1|}{n}\geq \epsilon\Big).
\end{equation}
Therefore by the converging giant and \eqref{eq: coupling thm1.1-2}, \eqref{eq: large comp is C1}, and \eqref{eq: obvious}
\begin{equation}\label{eq: A_delta}
\mathbb P\Big(\frac{\mathcal N}{n}\geq \epsilon\Big)\to\zeta(p).
\end{equation}
Similarly, be the coupling in Lemma~\ref{lm: coupling infection},
\[\mathbb P\Big(\zeta(p)-\epsilon\geq\frac{\mathcal N}{n}\leq\epsilon\Big)=\mathbb P_{G_n(p),\mathcal P_n}\Big(\zeta(p)-\epsilon\geq\frac{|C(v)|}{n}\geq \epsilon\Big).\]
Now, again, $\mathbb P_{\mathcal P_n}(\zeta(p)-\epsilon\geq\frac{|C(v)|}{n}\geq \epsilon)$ converges in probability  to the case that the random node $v$ is in $C_1$ and the same inequalities for $|C(v)|$ hold. However, by the converging giant condition we know that the probability that $\frac{|C_1|}{n}\leq\zeta(p)-\epsilon$ goes to zero. Therefore,
\begin{equation}\label{eq: Azeta}
\mathbb P\Big(\zeta(p)-\epsilon\geq\frac{\mathcal N}{n}\geq\epsilon\Big)\to 0.
\end{equation}

To finish the proof of part~\ref{Thm1.1-2}, let $f:[0,1]\rightarrow \mathbb R^+$ be any given continuous function bounded by $B>0$. Since it is continuous, for any $\epsilon'>0$ there exists $\delta$ such that $|f(x)-f(\zeta(p))|\leq \epsilon'$ for any $x\in(\zeta(p)-\delta,\zeta(p)+\delta)$, and also $|f(x)-f(0)|\leq \epsilon'$ for any $x\in[0,\delta)$. Then  to find the expectation of $f$, let $A_\delta$ be the event that $\frac{\mathcal N}{n}\leq \delta$,  let $A_{\zeta(p)}$ be the event that $|\frac{\mathcal N}{n}-\zeta(p)|\leq \delta$, and let $ A'$ be the event that non of $A_{\zeta(p)}$ and $A_{\delta}$ happen. Then
\begin{align*}
    \mathbb{E}\Big[f\big (\frac{\mathcal N}{n}\big)\Big]= \mathbb{E}\Big[f\big (\frac{\mathcal N}{n}\big)\mid A_\delta\Big]\mathbb{P}(A_\delta)+\mathbb{E}\Big[f\big (\frac{\mathcal N}{n}\big)\mid A_{\zeta(p)}\Big]\mathbb{P}(A_{\zeta(p)})+\mathbb{E}\Big[f\big (\frac{\mathcal N}{n}\big)\mid  A'\Big]\mathbb{P}( A').
\end{align*}
Similarly, we can write $\mathbb E f(\chi_p)$. Therefore,
\begin{align*}
    |\mathbb{E}\Big[f\big (\frac{\mathcal N}{n}\big)\Big]-\mathbb E f(\chi_p)|&\leq | \mathbb{E}\Big[f\big (\frac{\mathcal N}{n}\big)\mid A_\delta\Big]\mathbb{P}(A_\delta)-f(0)(1-\zeta(p))|
    \\&+\mid\mathbb{E}\Big[f\big (\frac{\mathcal N}{n}\big)\mid A_{\zeta(p)}\Big]\mathbb{P}(A_{\zeta(p)})-\zeta(p)f(\zeta(p))\mid
    +\mathbb{E}\Big[f\big (\frac{\mathcal N}{n}\big)\mid  A'\Big]\mathbb{P}( A')\\
    &\leq | f(0)\big(\mathbb P(A_\delta)-(1-\zeta(p))\big)+\epsilon|+| f(\zeta(p))\big(\mathbb P(A_{\zeta(p)})-\zeta(p)\big)+\epsilon|\\
    &\hfill +B\mathbb P(A').
\end{align*}
Now, the convergence of $\mathbb P(A_\delta)$ to $1-\zeta(p)$ proved in \eqref{eq: A_delta} and $\mathbb P(A_{\zeta(p)})$ to $\zeta(p)$ given in \eqref{eq: Azeta} implies the convergence of $\mathbb P(A')$ to $0$ and finishes the proof.
\end{proof}

\begin{proof}[Proof of Theorem~\ref{thm: outbreak size}]
The proof is a direct implication of Lemma~\ref{lm: outbreak size} and Corollary~\ref{thm: size of giant in expander}.
\end{proof}

\begin{remark}\label{remark:  lwc doesnt imply locality}
The local weak convergence of the graph sequence $\{G_n\}_{n\in\mathbb N}$ alone is not sufficient to read off the relative of size of the outbreak from its limit.  
 We can see this by comparing two graph sequences. For the first one, consider a graph sequence $\{H_n\}_{n\in\mathbb N}$ made by connecting two  $d$-regular random graph  components of size $\lfloor\frac n2 \rfloor$ with one edge. For the second sequence, consider $\{G_n\}_{n\in\mathbb N}$, where $G_n$ is a $d$-regular random graph of size $n$. The local weak limit of both $\{H_n\}_{n\in\mathbb N}$ and $\{G_n\}_{n\in\mathbb N}$ is a $d$-regular tree.
 
 However, in $H_n$ with probability at least $1-p$ the edge between the two $d$-regular component is removed, and the relative size of infection  will be roughly $\frac{\zeta(p)}{2}$. While there is also a chance that the two giant component becomes connected and result in a giant of the relative size $\zeta(p)$. So, in this case the size of an outbreak will have $3$ atoms instead.
\end{remark}
\subsection{Analysis of  Algorithm \ref{alg: forward process}}\label{sec: local alg}
Similar to the previous section, we will show that Algorithm \ref{alg: forward process} works  whenever the graph sequence has a converging giant. Again, we use the coupling in Lemma~\ref{lm: coupling infection} so we can focus only on the connected components after edge percolation. Then in graphs with a converging giant we know that, if a random node is not in the largest component, then the size of its component is $O_{k,\mathbb P}(1)$. So, if after the percolation or equivalently the infection process, the size of the connected component of the starting point is larger than $k$, then with high probability it lies in the giant component. Following that, a simple Chernoff bound gives the desired result.

\begin{lemma}[Generalization of Theorem~\ref{thm: local alg}]\label{lm: local alg}
Let $\{G_n\}_{n\in\mathbb{N}}$ be a graph sequence with a converging giant with the limit distributed as $\mu$ on $\mathcal G_*$. 
Then for any given $\epsilon>0$ there exist integers $k_{\epsilon}, q_\epsilon, N_\epsilon$,  such that for all $n\geq N_\epsilon$
    \[\mathbb{P}\big(|\widetilde{\mathcal N}(k_\epsilon,q_\epsilon)-\zeta(p)|\geq \epsilon\big)\leq \epsilon. \]
Here $\mathbb P$ denotes probabilities first over the randomness of $G_n$ and then the infection process and the sampled vertices contributing to $\widetilde{\mathcal N}(k_\epsilon,q_\epsilon)$.

\end{lemma}
\begin{proof}
Fix some $\epsilon$. First consider the case that $\zeta(p)>0$. By the assumption of converging giant, there exist $k_\epsilon$ and $N_\epsilon$ such that for $n\geq N_\epsilon$, $|\frac{|C_1|}{n}-\zeta(p)|\leq \frac{\epsilon}{3}$. Then by Proposition~\ref{prop: converging giant}
\begin{equation}\label{eq: lm2.6 eq 1}
    \mathbb P (|C(v)|\geq k_\epsilon,\text{ and } v\not\in C_1)\leq\frac{\epsilon}{3}.
\end{equation}

We need to make sure that when we sample $q$ nodes, their  $k$ neighborhoods are disjoint with high probability. So, running the SIR process in each neighborhood is independent.
Let $v_1,\ldots,v_q$ be the sampled initial nodes, and $X_1,\ldots, X_q\in\{0,1\}$ be the outputs of  the Algorithm \ref{alg: forward process}. Let $A_q$ be the event that the neighborhoods of $v_1,\ldots, v_q$ are disjoint. We claim that for large enough $n$
\begin{equation}\label{eq: bounded neighborhood}
    \mathbb P(A_q)\geq 1-q^2\frac{\log n}{n}-\frac{\epsilon }{3}.
\end{equation}
Then given $A_q$, infection process in $B_{k}(v_1),\ldots , B_{k}(v_q)$ are independent as well since the random variable drawn for edges are disjoint.   Therefore, 
\[\mathbb P(X_1,\ldots, X_q|A_q)=\prod_{i=1}^q\mathbb P(X_i\mid A_q),\]
and we will be able to use Chernoff bound to give concentration bounds. Before that, we prove \eqref{eq: bounded neighborhood}. 

Now, to prove \eqref{eq: bounded neighborhood}, define ${V}_{r,\Delta}$ as the set of set of nodes such that the maximum degree in their $r$ neighborhood is at most $\Delta$.
Since the limiting graph
$(G,o)$ is locally finite,  for all
$\epsilon>0$, there exists $\Delta<\infty$ and $N_\epsilon'<\infty$ such that for $n\geq N_\epsilon'$, with probability $1-\frac{\epsilon}{4}$ 
we have $\frac{|V_{k_\epsilon,\Delta}|}{n}\geq 1-\frac{\epsilon}{4}$. Let $E_\epsilon$ be the event that
$\frac{|V_{k_\epsilon,\Delta}|}{n}\geq 1-\frac{\epsilon}{4}$,
that $\phi(G_n,\epsilon)\geq \alpha$, and that $G_n$
has average degree at most $\bar d$.  Increasing $N_\epsilon$
if needed, we have that for $n\geq N_\epsilon$, $E_\epsilon$ has probability at least $1-\epsilon/3$. Given $E_\epsilon$, then by a union bound
\[\mathbb P (A_q \text{ does not happen})\leq  \mathbb P(A_q\text{ does not happen}\mid E_\epsilon)+\mathbb P(E_\epsilon)\leq q^2\frac{\Delta^k}{n}+\frac{\epsilon}{3}.\]
Since $q$ and $\Delta$ are independent from $n$, by choosing $n$ large enough we get \eqref{eq: bounded neighborhood}.

To finish the proof, recall that $\widetilde{\mathcal N}(k_\epsilon,q)=\frac{\sum_{i=1}^q X_i}{q}$. Also,  since $\zeta(p)>0$, for large enough $n$ we know $|C_1|\geq k$. With this and the definition of converging giant
\[\zeta(p)-\epsilon\leq \frac{|C_1|}{n}\leq \mathbb P (X_i=1).\]
Then combining the last equation with \eqref{eq: lm2.6 eq 1},
\[\zeta(p)-\epsilon\leq  \mathbb P (X_i=1)\leq \frac{|C_1|}{n}+\epsilon\leq \zeta(p)+2\epsilon.\]

Therefore, $|\mathbb E[X_i]-\zeta(p)|\leq 2\epsilon$. Moreover,
\[\mathbb P (A_q)\mathbb E(X_i|A_q)\leq \mathbb E (X_i)\leq \mathbb P (A_q)\mathbb E(X_i|A_q)+(1-P(A_q)).\]
Hence by \eqref{eq: bounded neighborhood},  $|\mathbb E[X_i|A_q]-\zeta(p)|\leq 3\epsilon$.
As a result, by Chernoff bound,
\begin{align*}
    \mathbb P(|\widetilde{\mathcal N}(k_\epsilon,q)-\zeta(p)|\geq t+3\epsilon \mid A_q)\leq e^{-qt^2}.
\end{align*}
Now we can  get a concentration bound without conditioning on $A_q$,
\begin{align*}
    \mathbb P(|\widetilde{\mathcal N}(k_\epsilon,q)-\zeta(p)|\geq t+3\epsilon)&\leq \mathbb P(|\widetilde{\mathcal N}(k_\epsilon,q)\zeta(p)|\geq t+3\epsilon \mid A_q)+1-\mathbb P (A_q)\\
    &\leq e^{-qt^2}+q^2\frac{\log n}{n}+\frac{\epsilon }{3},
\end{align*} 
which gives the desired result in the case that $\zeta(p)>0$. 

The case $\zeta(p)=0$ is similar, except that we can choose $k_\epsilon$ such that  $\mathbb P (|C(v)|\geq k_\epsilon)\leq \frac{\epsilon}{3}$. 
\end{proof}

\begin{proof}[Proof of Theorem \ref{thm: local alg}]
The proof follows from Corollary~\ref{thm: size of giant in expander} and Lemma~\ref{lm: local alg}.
\end{proof}

\begin{remark}
In the proof above, one can lift the assumption that the initial infected node is chosen uniformly at random. In fact, as long as each nodes is chosen with probability $O(1/n)$ then the relative size of an outbreak will still converge to $\zeta(p)$, since the size of the largest component in $G(p)$ converges to $\zeta(p)$ and the Chernoff bounds in the proof above follows similarly and the relative size of the largest component still .

However, the probability of an outbreak might change. When the initial infection happens proportional to some local feature of nodes such as the degree, we can first draw uniform random samples,  calculate their local features (degrees), and then do rejection sampling to get the probability of an infection when one infects by choosing vertices according to the local feature (degree). When the local feature is degree and the degrees have a finite second moment, then by applying Theorem 2.23 \cite{RemcoVol2} we can see the local limit of the above rejection sampling procedure exists and gives the right probability of an outbreak.
It is an interesting question to characterize other local features that starting the infection with respect to them keeps the probability of an outbreak a local property, i.e., it can be derived from the rejection sampling modification of Algorithm~\ref{alg: forward process}.
\end{remark}

\section{Proofs of Applications to Motif-Based Graph Models}\label{sec: proofs application}
{
This section is on the proof of applications to motif-based models. 
For that purpose, we need to check 1) local weak convergence, 2) large-set expansion of the sequence, and 3) continuity of the survival probability in the limit.  
We go over the proofs of these properties 
in the same order. 
Then in Section~\ref{sec: motif based}, we show Theorems~\ref{thm: outbreak size} and \ref{thm: local alg} hold for the motif-based configuration model and the preferential attachment models. 
}
\subsection{Locality of Motif-Based graphs}

{We start by proving the locality stated in  Lemma~\ref{lm: LWC of MRG}. Let us first state the precise definition of local weak convergence.  {For this purpose, we define a function $f:\mathcal{G}_*\longrightarrow \mathbb{R}$ to be $k$-local if $f(G,o)$ depends only on the $k$-neighborhood of $o$ in $G$.}
Given a sequence of (possibly) random graphs $\{G_n\}_{n\in\mathbb{N}}$, and a {(non-random)} probability $\mu$ on $\mathcal{G}_*$, we say  $G_n$ converges in probability \emph{in the local weak sense} to $\mu$ if {for $k<\infty$,} and any bounded {$k$-local} function $f:\mathcal{G}_*\longrightarrow \mathbb{R}$, 
\begin{equation}\label{eq: lwc definiiton}
\mathbb{E}_{\mathcal{P}_n}[f]\overset{\mathbb{P}}{\to}\mathbb{E}_\mu[f],
\end{equation}
where  $\mathcal P_n$ {is the uniform distribution over random vertices in $G_n$ and $\mathbb{E}_{\mathcal{P}_n}[f]$ is a shorthand for
$\mathbb{E}_{v_n\sim \mathcal{P}_n}[f(G_n,v_n)]$}.%
\footnote{{Alternatively, local weak convergence can be defined by introducing a metric on  $\mathcal{G}_*$
and then using continuous functions and weak convergence with respect to this metric \cite{Aldous2004, benjamini2001}; hence the name local weak convergence.  The equivalence of the two is proven, e.g.,} in Section 2 of \cite{RemcoVol2}).}}

To prove Lemma~\ref{lm: LWC of MRG} and hence local weak convergence of a general class of motif-based models, we will use the second moment method to prove \eqref{eq: lwc definiiton}. The key idea is to map functions on the motif-based graphs to the ones on the external graphs by averaging over all vertices in a motif and then use the local weak convergence of the external graphs. 
But before, we need a technical result, which shows the expected size of a motif converges in probability.

\begin{proposition}\label{prop: internal 2}
Let  $\{G_n\}_{n\in\mathbb N}$ be a motif-based graph sequence satisfying Conditions \ref{cond: internal} and \ref{cond: external}. 
Let $\mathbb E_n$ denote the probability distribution of a uniform random motif among all $n$ motifs of $G_n$, then there exist $\bar v<\infty$,
\[\mathbb E_n[v(M)]\overset{\mathbb P}{\to}\bar v.\]
Here $\overset{\mathbb P}{\to}$ denotes convergence in probability with respect to the possible randomness of the degree distribution of the external graph, as well as the randomness from assigning motifs iid at random according to 
$\mathcal M_d$.
\end{proposition}
\begin{proof}

Let 
$\tilde v
$ be the expectation of $ \mathbb E_n[v(M)]$ conditioned on $G^{ext}_n$.  We can write this expectation as
$$
\tilde v=\frac 1n\sum_{i=1}^n v_{d_i},
$$
where $d_i$ is the degree of vertex $i$ in $G^{ext}_n$,
and
$
v_d=\mathbb E_{\mathcal M_d}[v(M)]
$.  On the other hand, the expectation of the square of 
$ \mathbb E_n[v(M)]$ conditioned on $G^{ext}_n$ is equal to
$$
\frac 1{n^2}\sum_{i\neq j}^n v_{d_i}v_{d_j}+
\frac 1{n^2}\sum_i \mathbb E_{\mathcal M_{d_i}}[v^2(M)]
\leq \tilde v^2+\frac{S_{\max}^2}n.
$$
Thus, by the second moment method
$$
\mathbb P(|\mathbb E_n[v(M)]-\tilde v|\geq\epsilon)
\leq \frac{S_{\max}^2}{n\epsilon^2}
$$
where the probability is over the random motif assignments.  Combined with the fact that $\tilde v$
converges in probability
by the assumed weak local convergence of $G^{ext}_n$, this proves the proposition.
\end{proof}

\begin{proof}[Proof of Lemma \ref{lm: LWC of MRG}]
To prove the local weak limit of the motif-based graph, we need to define the probability distribution of the limiting graph $\mu$, which we want to construct using $\mu^{ext}$. For any graph $H$, with the abuse of notation we use $\mathcal M(H)$ as the graph constructed by replacing each vertex $v$ in $H$ with a motif drawn from the distribution $\mathcal M_{d_v}$ (here, $d_v$ is the degree of $v$). 

We need to make sure the neighborhood of a uniform random node in $G_n$ converges to  the neighborhood of the graph $(G,o)\sim\mu$. For that purpose, we need to re-weight 
nodes {in the external network} based on their motif sizes. So, let $\bar\mu ^{ext}$ be the weighted probability measure on the external graphs, i.e., we draw $(G^{ext},o)$ proportional to $\mathbb E_{\mathcal M_{d_o}}[v(M)]\mu^{ext}(G^{ext},o)$, where $d_o$ is the degree of the root $o$. 
Let $\mu$ be the probability measure that replaces each node $v$ of $(G^{ext},o)\sim \bar\mu^{ext}$ with a motif drawn from $\mathcal M_{d_v}$, and then chooses a uniform random node $o'$ from $M_o$ (the motif of the external root) to be the new root of the graph.

Fix an integer $k>0$ and  a {$k$-local} bounded  function $f$ on
$\mathcal G_*$. {Let $n'=\sum_{i=1}^n v(M_i)$
be the number of nodes in the motif-based graph $G_n$, and let
$\mathcal P_{n'}$ denote the distribution of a uniform random node among the $n'$ nodes in the motif graph $G_n$.}
For the proof of local weak convergence,
we need to show that 
\begin{equation}\label{eq: lwc int-ext}
\mathbb E_{\mathcal P_{n'}}[f(G_n)] \overset{\mathbb P}{\to} \mu\big( f(G,o)\big),
 \end{equation}
 where, with a slight abuse of notation, we used $\mathbb E_{\mathcal P_{n'}}[f(G_n)]$ for $\mathbb E_{v\sim\mathcal P_{n'}}[f(G_n,v)]$.

We will use the second moment method to prove  \eqref{eq: lwc int-ext}. 
Given an external graph $G_n^{ext}$, we use $\mathbb E_{\mathcal M({G_n^{ext})}}$ for {expectations with respect to a random graph $G_n$ drawn from $\mathcal M(G_n^{ext})$. }
For the first moment we need to show,
\begin{equation*}
    \mathbb E_{\mathcal M(G_n^{ext})}\big[\mathbb E_{\mathcal P_{n'}}[f({G_n})] \big]\overset{\mathbb P}{\to} \mathbb E_{\mathcal M(G^{ext})}\big[\mu\Big(f(G,o)\Big)\big],
\end{equation*}
where the convergence in probability is over all possible randomness of $G_n^{ext}$. Note that 
\begin{align*}
    \mathbb E_{\mathcal M(G_n^{ext})}[\mathbb E_{\mathcal P_{\sum_{i\in[n]} v(M_j)}}[f(G_n)]]&=\frac{1}{n}\sum_{i\in[n]} \mathbb E_{\mathcal M(G_n^{ext})}[\frac{n}{\sum_{i\in[n]} v(M_j)}\sum_{v\in M_i}f(G_n,v)].
\end{align*}
By Proposition \ref{prop: internal 2}, $\frac{n}{\sum_{i\in[n]} v(M_j)}$ converges in probability to $\frac{1}{{\bar v}}$. So, if we define 
$$\bar f(G_n^{{ext}},i)=\mathbb E_{\mathcal M(G_n^{{ext}})}[\frac{1}{\bar v}\sum_{v\in M_i}f(G_n,v)],$$
then
\begin{align*}
    \mathbb E_{\mathcal M(G_n^{{ext}})}[\mathbb E_{\mathcal P_{\sum_{i\in[n]} v(M_j)}}[f(G_n]]&-\frac{1}{n}\sum_{i\in[n]} \bar f(G_n^{{ext}},i)\\
    &=\frac{1}{n} \mathbb E_{\mathcal M(G_n)}\Big(\big(\frac{n}{\sum_{j\in[n]} v(M_j)}-\frac{1}{{\bar v}}\big) \big(\sum_{v\in M_i}f(G_n^{{ext}},v)\big)
    \Big).
\end{align*}
Fix some $\epsilon>0$, and let $B$ be the bound on $f$. Then by Proposition \ref{prop: internal 2}, $$ \mathbb P\big(|\frac{n}{\sum_{j\in[n]} v(M_j)}-\frac{1}{{\bar v}}|\geq \frac{\epsilon}{B}\big)\to 0.$$
Therefore,
 \begin{align*}
   \mathbb P\Big(   \frac{1}{n} \mathbb E_{\mathcal M(G_n)}\Big(|\frac{n}{\sum_{j\in[n]} v(M_j)}-\frac{1}{{\bar v}}| \big(\sum_{v\in M_i}f(G_n,v)\geq \epsilon\big)    \Big)\to 0.,
 \end{align*}
 and as a result,
\begin{align*}
 \mathbb P\Big(  | \mathbb E_{\mathcal M(G_n)}[\mathbb E_{\mathcal P_{\sum_{j\in[n]} v(M_j)}}[f(G_n]&-\frac{1}{n}\sum_{i\in[n]}\bar f(G_n^{{ext}},i)|\geq \epsilon\Big)\to 0.
\end{align*}

So, to prove the first moment it is enough to show,
\begin{equation}\label{eq: first momemnt motif}
   \mathbb E_{\mathcal P_n} \bar f(G_n^{ext})\overset{\mathbb P}{\to} \mathbb E_{\mathcal M(G^{ext})}\big[\mu\Big(f(G,o)\Big)\big].
\end{equation}
{Note that $\bar f$ } is bounded because
$$\bar f(G_n^{{ext}},i)=\mathbb E_{\mathcal M(G_n)}[\frac{1}{\bar v}\sum_{v\in M_i}f(G_n,v)]\leq \frac{B S_{\max}}{\bar v}.$$
{Furthermore, if $j$ has distance  $k$ from $i$,
then any node in $M_i\subset G_n$ has distance at least $k$ from any node in $M_j\subset G_n$, showing that 
$\bar f(G_n^{{ext}},i)$ depends only on the $k$-neighborhood of $i$ in $G_n^{{ext}}$.
Thus 
 $\bar f$ is  a $k$-local, bounded function, and} 
local weak convergence of the external graphs implies that
\begin{align*}
      \mathbb E_{\mathcal P_n} \bar f(G_n^{{ext}})\overset{\mathbb P}{\to}\mu^{ext}\Big(\bar f(G^{ext},o)\Big)=\mu^{ext}\Big(\mathbb E_{\mathcal M(G^{ext})}[\frac{1}{ \bar v}\sum_{v\in M_o}f(G,v)]\Big),
\end{align*}
which  implies \eqref{eq: first momemnt motif}.


For the second moment, we first note that {with high probability} two uniform random motifs in $G_n^{ext}$ have disjoint $k$-neighborhoods.   In fact, if $dist_{G_n^{ext}}(i,j)$ is the distance of motifs $i$ and $j$ in $G_n^{ext}$, then  by the local weak convergence of $G_n^{ext}$,  Corollary 2.19 in \cite{RemcoVol2} implies that $dist_{G_n^{ext}}(i,j)\overset{\mathbb P}{\to}\infty$.
Next, we show the second moment argument follows from the large distances of motifs. We can write,
\begin{align*}
    \mathbb E_{\mathcal M(G_n^{ext})}[(\sum_{v\in V(G_n)}f(G_n,v))^2]&=\mathbb E_{\mathcal M(G_n^{ext})}[\sum_{i,j\in[n]}(\sum_{v\in M_i} f(G_n,v))(\sum_{u\in M_j} f(G_n,u))]\\
    &=\sum_{i,j\in[n]}\mathbb E_{\mathcal M(G_n^{ext})}[(\sum_{v\in M_i} f(G_n,v))(\sum_{u\in M_j} f(G_n,u))].
\end{align*}
Then if we divide the outer sum into two cases based on whether the distance of $M_i$ and $M_j$ is larger or smaller than $k$, we can use the fact 
{when it is larger than $k$, the expectation over
the choice of motifs factors, giving }
\begin{align*}
    \mathbb E_{\mathcal M(G_n^{ext})}&[(\sum_{v\in V(G_n)}f(G_n,v))^2]\\
    &=\sum_{i,j\in[n], dist_{G_n^{ext}}(i,j)> k}\mathbb E_{\mathcal M(G_n^{ext})}[(\sum_{v\in M_i} f(G_n,v))]\mathbb E_{\mathcal M(G_n^{ext})}[(\sum_{u\in M_j} f(G_n,u))]\\&+
    \sum_{i,j\in[n], dist_{G_n^{ext}}(i,j)\leq k}\mathbb E_{\mathcal M(G_n^{ext})}[(\sum_{v\in M_i} f(G_n,v))(\sum_{u\in M_j} f(G_n,u))].
\end{align*}
{Here we used} that the distance of two nodes in $G_n$ is at least the distance of their corresponding motifs in the external graph.
Then we choose $n$ large enough such that for random motifs (proportional to their sizes) their distance is less than $k$ with probability at most $\epsilon$,
\begin{align*}
  \sum_{i,j\in[n], dist_{G_n^{ext}}(i,j)\leq k}\mathbb E_{\mathcal M(G^{ext}_n)}[(\sum_{v\in M_i} f(G_n,v))(\sum_{u\in M_j} f(G_n,u))]\leq \epsilon n^2S_{\max}^2B^2,
\end{align*}
where $B$ is the bound on $f$.
Therefore, since $\epsilon$ was arbitrary
\begin{align*}
    \mathbb E_{\mathcal M(G_n^{ext})}[(\frac{1}{n}\sum_{v\in V(G_n)}f(G_n,v))^2]&\overset{\mathbb P}{\to} \mathbb E_{\mathcal M(G_n^{ext})}[\frac{1}{n}\sum_{v\in V(G_n)}f(G_n,v)]^2,
\end{align*}
where the convergence in probability is over the randomness of the external graph $G_n^{ext}$.

On the other hand,
\begin{align*}
    \mathbb E_{\mathcal M(G_n^{ext})}[\frac{1}{(\sum_{i\in[n]} v(M_i))^2}(\sum_{v\in V(G_n)}f(G_n,v))^2]&\\= \frac{1}{n^2}\sum_{i\in[n]} \mathbb E_{\mathcal M(G_n^{ext})}[&\frac{n^2}{(\sum_{i\in[n]} v(M_i))^2}(\sum_{v\in M_i}f(G_n,v))^2].
\end{align*}
Again, using Proposition~\ref{prop: internal 2},
\begin{align*}
    \mathbb E_{\mathcal M(G_n^{ext})}[\frac{1}{(\sum_{i\in[n]} v(M_i))^2}(\sum_{v\in V(G_n)}f(G_n,v))^2]&   \overset{\mathbb P}{\to}  \mathbb E_{\mathcal M(G_n^{ext})}[(\frac{1}{n}\sum_{v\in V(G_n)}f(G_n,v))^2],
\end{align*}
and
\begin{align*}\mathbb E_{\mathcal M(G_n^{ext})}[(\frac{1}{n}\sum_{v\in V(G_n)}f(G_n,v))]^2
&   \overset{\mathbb P}{\to}      \mathbb E_{\mathcal M(G_n^{ext})}[\frac{1}{(\sum_{i\in[n]} v(M_i))}(\sum_{v\in V(G_n)}f(G_n,v))]^2,
\end{align*}
which gives the convergence of the second moment, i.e.,
\begin{align*}
 \mathbb E_{\mathcal M(G_n^{ext})}\big[\frac{1}{(\sum_{i\in[n]} v(M_i))^2}\big(\sum_{v\in V(G_n)}f(G_n,v)\big)^2\big]
&   \overset{\mathbb P}{\to}      \mathbb E_{\mathcal M(G_n^{ext})}\big[\frac{1}{(\sum_{i\in[n]} v(M_i))}\big(\sum_{v\in V(G_n)}f(G_n,v)\big)\big]^2.
\end{align*}
As a result,
\[\frac{\var_{\mathcal M(G_n^{ext})}\big(\mathbb E_{\mathcal P_{\sum_{i\in[n]} v(M_i)}}[f]\big)}{\mathbb E_{\mathcal M(G_n^{ext})}\big(\mathbb E_{\mathcal P_{\sum_{i\in[n]} v(M_i)}}[f]\big)^2}\overset{\mathbb P}{\to}0.\]
Then by Chebyshev's inequality,
\begin{equation}
    \frac{\mathbb E_{\mathcal P_{\sum_{i\in[n]} v(M_i)}}[f]}{\mathbb E_{\mathcal M(G_n^{ext})}\big(\mathbb E_{\mathcal P_{\sum_{i\in[n]} v(M_i)}}[f]\big)}
\overset{\mathbb P}{\to}1.
\end{equation}
So  by the convergence of the first moment \eqref{eq: first momemnt motif}, we get the desired result in \eqref{eq: lwc int-ext}.

\end{proof}

\subsection{Expansion of Motif-Based Graphs}

Next we show that motif based graphs inherit large-set expansion from their external graph structure if the motif sizes are bounded.
Edge percolation on motifs can be viewed as an analog of vertex percolation on the external graph. So, in order to prove edge-expansion of motif-based graphs, we first prove that large-set edge expansion of the external graph $\{G_n^{ext}\}_{n\in\mathbb N}$ implies  large-set \emph{vertex} expansion.

For that purpose, we first define large-set vertex expansion similar to the definition of  large-set edge expansion. Given a graph $G=(V,E)$ and 
a set $A\subset V$, let $\partial_{out}$ be the {set} of vertices in $V\setminus A$ that have at least one neighbor in $A$, {and let $\delta_{out}(A)=|\partial_{out}|$. For} 
$\epsilon<1/2$, we then define
\begin{equation}
\label{eq: vertex expansion}
\phi_{out}(G,\epsilon)=\min_{A\subset V: \epsilon |V|\leq |A|\leq |V|/2}
\frac{\delta_{out}(A)}{|A|}.
\end{equation}
Call a graph $G$ an $(\alpha,\epsilon,\bar d)$ large-set vertex expander if the average degree of $G$ is at most $\bar d$ and $\phi_{out}(G,\epsilon)\geq \alpha$.  A sequence of possibly random graphs $\{G_n\}_{n\in\mathbb{N}}$ is  called a large-set \emph{vertex} expander sequence with bounded average degree,  if there exists $\bar d<\infty$ and $\alpha>0$ such that for  any $\epsilon\in (0,1/2)$ the probability that $G_n$ is an  $(\alpha,\epsilon,\bar d)$ large-set vertex expander goes to $1$ as $n\to\infty$.

In the following, we prove {that} large-set edge expansion together with local weak convergence is enough to imply  large-set vertex expansion. Note that this is straight forward when the maximum degree is bounded. Here, the main observation is that graphs with local weak limits have a tight local neighborhood, in the sense that the degree of nodes in the neighborhood of a random edge is bounded with high probability.

\begin{lemma}\label{lm: vertex expansion}
Let  $\mu$ be a probability distribution on $\mathcal G_*$, and let $\{G_n\}_{n\in\mathbb N}$ be a sequence of $(\alpha,d)$ large-set edge expanders  with the local weak limit $(G,o)\sim \mu$. Then $\{G_n\}_{n\in\mathbb N}$  is a large-set vertex expander as well.
\end{lemma}
\begin{proof}
Assume to the contrary that there exists some $\epsilon>0$ such that there exists an infinite subsequence $\{\tilde G_{n'}\}$ with $\phi(\tilde G_n',\epsilon)\geq \alpha$ but $\phi_{out}(\tilde G_n',\epsilon)=o(1)$.  Then for any $\delta>0$ and large enough $n$ we have $\phi_{out}(\tilde G_n,\epsilon)<\delta$. Let $A_n\subset V(G_n)$ be a set with bad vertex expansion, i.e., $\epsilon n\leq |A_n|\leq n/2$ and $\delta_{out}(A_n)\leq \delta|A_n|$.  On the other hand, by the edge expansion of $A_n$, the number of edges coming out of $A_{out}=\partial_{out}(A_n)$ is at least $\alpha|A_n|$.
So, the average number of edges coming out of a node from $A_{out}$ is {at least}
$\alpha/\delta$. Therefore,
if we choose a uniform random edge from the edges coming out of $A_{out}$ with probability at least $1/2$ its endpoint has degree larger than $\alpha/(2\delta)$.
To see this,  
let $d_1\geq d_2\geq \cdots\geq d_{m'}$ be the number of edges coming out of each vertex in $A_{out}$, and let $\bar d=\frac{\sum_{i=1}^{m'}d_i}{m'}$ be the average of them. Also, let $s$ be the index that $d_1\geq\cdots\geq d_s\geq  \frac{\bar d}{2}\geq d_{s+1}\geq\cdots \geq d_{m'}.$ Now if we choose a random edge it will be incident to one of the vertices of degree larger than $\bar d/2$ with probability $\frac{d_1+\cdots+d_s}{d_1+\cdots+ d_{m'}}$. So to prove our claim, we need to show that $d_1+\cdots+d_s\geq d_{s+1}+\cdots+d_{m'}$. Assume, to the contrary that this is not true. Then we have 
\[m'\bar d=\sum_{i=1}^{m'} d_i\leq 2 (d_{s+1}+\cdots+d_{m'})\leq (m'-s)\bar d,\]
which is a contradiction.
Note that since $\delta$ was arbitrary,  we can make $\alpha/2\delta$ as large as we want. We will see this will result in a contradiction. 

{To see this, pick} a uniform random edge from $G_n$. With probability at least $\alpha\frac{|A_n|}{{\bar d} n}\geq \frac{\alpha\epsilon}{{\bar d}}$ it will be in the boundary of $A_{out}$. Hence, with probability at least $\frac{\alpha\epsilon}{{2\bar d}}$  the degree of the endpoint of the edge is at least $\frac{\alpha}{2\delta}$.  
On the other hand, by Theorem 2.23 in \cite{RemcoVol2}, since the average degree is bounded and hence the degree of a random node  in $G_n$ is uniformly integrable,  the probability that the neighborhood of a uniformly random edge $e$  is isomorphic to a finite graph $H\in \mathcal G_*$ converges in probability to $\frac{\mu[d_o\mathbf 1_{B(G,o)\sim H}]}{\mu[d_o]}$, where $d_o$ is the degree of the root in the limit. 
Since the local neighborhoods in $\mathcal G_*$ are finite, {we conclude that} exists $\Delta$ such that {for $n$ large enough,} with probability at least $1-\frac{\alpha\epsilon}{{4} \bar d}$, the degree of the neighbors of {a random edge in $G_n$}
is at most $\Delta$. {If we choose $\delta$ in such a way that $\Delta< \frac{\alpha}{2\delta}$, we get a contradiction.}

\end{proof}

Next we prove {that} large-set edge-expansion of motif-based graphs  follows form  large-set vertex expansion of the external graphs. 
\begin{proposition}
\label{prop: motif expansion}
Let $\{G_n\}_{n\in\mathbb N}$ be a sequence of motif-based graphs satisfying Conditions~\ref{cond: internal} and \ref{cond: external}. Assume the external sequence $\{G_n^{ext}\}_{n\in\mathbb N}$ is a large-set expander with bounded average degree. Then  $\{G_n\}_{n\in\mathbb N}$ is a large-set expander with bounded average degree as well.
\end{proposition}
\begin{proof}
We first show the large-set edge-expansion of  $\{G_n\}_{n\in\mathbb N}$ and then we bound the average degree.
Fix some $\epsilon\in(0,1/2)$ and a set $S\subset V(G_n)$ with $|S|\geq \epsilon n$. Since there are $n$ motifs in $G_n$ each of size at most $S_{\max}$, $|S|\geq \frac{\epsilon}{S_{\max}} |V(G_n)|$. Note that some motifs are divided between $S$ and $V(G_n)\setminus S$, and some motifs are either completely in $S$ or $V(G_n)\setminus S$. Let $M(S)$ be the set of motifs that are completely in $S$, i.e.,  for each motif $M\in M(S)$ we have $V(M)\subset S$. Also, let $\tilde M(S)$ be the set of motifs that have at least one vertex in $S$, so $M(S)\subseteq \tilde M(S)$. 

Since the size of each motif is at most $S_{\max}$, we have $|\tilde M(S)|\geq\frac{|S|}{S_{\max}}\geq \frac{\epsilon}{S_{\max}}n$. Now, consider two cases. First, assume that the number of motifs divided between $S$ and $V(G_n)\setminus S$ (motifs in $\tilde M(S)\setminus M(S)$) is at least $\frac{|S|}{2S_{\max}}$ ($\geq\frac{\epsilon}{2S_{\max}}n$). Since, each motif is a connected graph, there are at least
{$\frac{|S|}{2S_{\max}}$} 
edges between $S$ and $V(G_n)\setminus S$.  Therefore, the
{edge expansion of $S$ in $G_n$ is at least $\frac{1}{2S_{\max}}$.}

In the second case, $|\tilde M(S)\setminus M(S)|\leq \frac{|S|}{2S_{\max}}$ and as a result, $|M(S)|\geq \frac{|S|}{2S_{\max}}$. Here, we use large-set vertex expansion of external graphs. By Lemma~\ref{lm: vertex expansion}, there exists some $\alpha>0$ independent of $\epsilon$ such that with high probability $\phi_{out}(G_n^{ext},\frac{\epsilon}{2S_{\max}})\geq \alpha$. Note that all motifs in $\partial_{out}(M(S))$ are either divided between $S$ and $V(G_n)\setminus S$ or they are completely in $V(G_n)\setminus S$ but connected to a motif in $M(S)$. In both cases, each of the motifs contribute at least one edge to $e(S,V(G_n)\setminus S)$.  Therefore the edge expansion of $S$ in $G_n$ is whp at least $\alpha$. As a result, the sequence $\{G_n\}_{n\in\mathbb N}$ is a large-set edge expander.

Now it remains to bound the average degree. For  the graph $H$, let $d_{avg}(H)$ be the average degree. Let $\{M_i=(F_i,\mathbf d_i)\}_{i=1}^n$ be the motifs of $G_n$. Then we can compute the average degree by summing up the internal and external edges,
\begin{align*}
    d_{avg}(G_n)=&\frac{\sum_{i=1}^n v(M_i)d_{avg}(F_i)}{\sum_{i=1}^nv(M_i)}+\frac{n d_{avg}(G_n^{ext})}{\sum_{i=1}^nv(M_i)}&\\
    \leq & \frac{S_{\max}^2+\bar d}{\frac{\sum_{i=1}^n v(M_i)}{n}}
    &\text{( Condition~\ref{cond: internal}),}
\end{align*}
where $\bar d$ is the average degree of the external graph. 
Since each motif has size at least $1$ the average degree is upper bounded by $S_{\max}^2+\bar d$. 
\end{proof} 
As a direct implication of large-set expansion and local weak limit of motif-based graphs we get Corollary~\ref{thm: motif based general}.
\begin{proof}[Proof~\ref{thm: motif based general}]
By Proposition~\ref{prop: motif expansion} the sequence $\{G_n\}_{n\mathbb N}$ is a large-set edge expander, and by Lemma \ref{lm: LWC of MRG} the sequence has a local weak limit. Therefore, we can apply Theorems~\ref{thm: outbreak size} and \ref{thm: local alg} to motif-based graph with a smooth limit.
\end{proof}

\subsection{Continuity of Survival Probability}\label{sec: motif continuous}
Finally, we need to prove the smoothness of the local weak limit for the motif-based preferential attachment and configuration model. The limit of the preferential attachment and configuration model without motifs
are {in} a subclass of multi-type branching processes, called threshold regular branching processes, {for which} the continuity of  the survival probability  was proven in Lemma 6.2. in \cite{abs2021locality}. We extend the  results of \cite{abs2021locality} to motif-based graphs where the external sequence $\{G_n^{ext}\}_{n\in\mathbb N}$ is a threshold regular branching processes.

At high level, in  the multi-type branching processes {considered in
\cite{abs2021locality}, each vertex has a label in the continuum interval}  $[0,1]$, and {the distribution of} the number and types of children of a node $v$ conditioned on its label is independent of the tree above $v$.  {The condition of threshold regularity then requires} that for $\zeta(p)>0$, the multi-type  branching process 
 conditioned on the  label of the root ($X_0\in [0,1]$)
survives with a probability that is  bounded away from zero uniformly in $X_0$. See the precise definition in Section 6.1. of \cite{abs2021locality}. Next, we give an overview of their proof to extend their result.

The proof starts with the observation that $\zeta(p)$ is the point-wise limit of  a sequence of functions $\tilde \zeta_k(p)$, defined as the probability {of the event $A_k$} that after percolation, the component of the root extends to at least distance $k$ from the root.
Continuity of $\zeta$ {is then established by} finding a uniform upper bound on $\frac{d}{dp}\tilde\zeta_k(p)$.
{This in turn is achieved with the help}  of  Russo's formula \cite{Russo1981OnTC}. To state it,  define an edge  
$e$ to be pivotal if exactly  one of $G(p)\cup\{e\}$ and $G(p)\setminus\{e\}$ {is in $A_k$.}  The Margulis--Russo formula  \cite{Russo1981OnTC} then says that
\begin{equation}\label{Russo-1}
 \frac d{dp}\tilde\zeta_k(p)=\sum_{e\in [m]}\mathbb P_{p}(e\text{ is pivotal}).
\end{equation} 
{If $e$ is pivotal, and $e$ is an edge in $G(p)$, the edge 
is called a 
$k$-bridge.  Since edges appear with probability $p$, }
 the right hand side is $1/p$ times the {expected}
 number of $k$-bridges in $G(p)$.  
Thus our goal is to find a uniform bound on the number of $k$-bridges, which can be reduced to bounding  the number of $\infty$-bridges, {defined as} edges such that  their removal turns the infinite component of the root into a finite component. 

{In \cite{abs2021locality}, the upper bound on the number of $k$-bridges is established by using threshold regularity to prove}
the probability that an edge is an $\infty$-bridge decays as the distance from the root increases.  {But unfortunately,}
threshold regularity  does not necessarily hold when we add motifs 
to the tree. However, we are still able to extend their proof, by bounding the number of $k$-bridges in the motif-based graphs, which is at most the number of $k$-bridges in the external graph multiplied by the maximum household size. The complete proof appears in Appendix \ref{sec: appendix continuity}.

\begin{lemma}\label{lm: continuity motif graph}
Let $\{G_n\}_{n\in\mathbb N}$ be a sequence of motif-based graphs with the limit $\mu\in\mathcal G_*$ satisfying the Condition~\ref{cond: internal} and ~\ref{cond: external}. Also, assume that $\mu_{ext}$, the limit of the external graph, is threshold regular (Definition~6.2 in \cite{abs2021locality}). Then $\zeta(p)=\mu(\mathbb P_{G(p)}\big((C,o)=\infty\big)$ is continuous for all $p>p_c(\mu^{ext})$, {where $p_c(\mu^{ext})$ is the percolation threshold for the limit of the external graph.}
\end{lemma}

\subsection{Proofs of the Motif-Based Configuration Model and Preferential Attachment}\label{sec: motif based}

Finally, we prove Corollaries~\ref{thm: motif based PA} and \ref{thm: motif based CM}, by noting that continuity of the  survival probability in the limit is already implied by Lemma~\ref{lm: continuity motif graph} for motif-based preferential attachment and motif-based configuration models.
\begin{proof}[Proof of Corollary~\ref{thm: motif based PA}]
Motif-based preferential attachment graphs, are large-set expanders by the expansion of preferential attachment models shown in Lemma~\ref{lm: PA- expansion + continuity} and Proposition~\ref{prop: motif expansion}. Their local weak limit (which exists by Lemma~\ref{lm: LWC of MRG}) has a continuous survival probability by Lemma~\ref{lm: continuity motif graph} above and the fact that the external graph has a smooth survival probability at $p_c(\mu^{ext})=0$ by Corollary 6.8 in \cite{abs2021locality}. {The continuity of the motif-based graph at $p=0$ follows from the fact that
$\zeta(p)\leq \zeta^{ext}(p)$, because by percolating edges in a motif the corresponding super node in the external graph will be removed when the motif becomes disconnected. This event happens with a positive probability because the size of the motifs are bounded by some constant $S_{\max}$. Therefore since $\zeta^{ext}(p)\to 0$ as $p\to 0$, we get that $\zeta(p)\to 0$ for $p\to 0$ as well.}

As a result  Theorems~\ref{thm: outbreak size} and \ref{thm: local alg} applies to motif-based preferential attachment graphs.
\end{proof}
\begin{proof}[Proof of Corollary~\ref{thm: motif based CM}]
Similar to the motif-based preferential attachment model, the expansion follows from  the expansion of the external graph proved in Lemma 12 in \cite{abdullah2012cover} and Proposition~\ref{prop: motif expansion}. The continuity follows from Lemma~\ref{lm: continuity motif graph} and noting that the limit of the external graphs is a two-stage  branching process(by Theorem 4.5 in \cite{RemcoVol2}) {and hence has a continuous survival probability ( see e.g., \cite{broman2008survival}).}

We only need to check the continuity at or below $p_c^{ext}$, the percolation threshold for the external graph. Similar to the argument in Corollary~\ref{thm: motif based PA}
$\zeta(p)\leq \zeta^{ext}(p)$. Therefore, $\zeta(p)=0$ for all $p<p_c^{ext}$.
Moreover, since $\zeta^{ext}(p)\to0$ as $p\to p_c^{ext}$ for the configuration model (see e.g., \cite{broman2008survival}), then  $\zeta(p)\to0$ for $p\to p_c^{ext}$ as well.
As a result, all the conditions needed for Theorems~\ref{thm: outbreak size} and \ref{thm: local alg} hold and their result is applicable to motif-based configuration models.

\end{proof}

\section{Future Work}
A natural next step is to prove that our results hold for the general SIR model, where the recovery time is drawn from a distribution.
Intuitively, Algorithm~\ref{alg: forward process} is related to the probability of an outbreak starting from a uniform random node. For a general SIR model this is not the case and one can define a  backward infection process (analogous to the forward process described in Algorithm~\ref{alg: forward process}) to compute the relative size of an outbreak. The backward process at each step, finds which of the undiscovered neighbors could have transmitted the disease to the original node.
Note that if a node is infected in an outbreak, then its backward infection must be large and hence it will reach $k$ other nodes. So, the possible next step would be to characterize graphs such that the backward and forward processes describe the size and probability of an outbreak for a generalized SIR model.

Further, we conjecture that the assumption on continuity of $\zeta(p)$ in all of our results can be lifted, i.e., the local weak limit of a sequence of large-set expanders has a continuous survival probability. 

\bibliographystyle{aer}
\bibliography{ref.bib}
\newpage

\appendix


\section{Proof of Lemma~\ref{lm: continuity motif graph}} \label{sec: appendix continuity} 
To give the formal proof, first  recall the definition of threshold regular trees from \cite{abs2021locality}. Using their  terminology, types of vertices in a multi-type branching process $G$ are labelled by a label $x\in[0,1]$. The distribution on the types and the number of children of a vertex {of type $x$} only depends on $x$ and can be as general as desired. 
Let $\mu_x$ be the distribution over trees generated starting from a vertex of type $x$, and let $G_x$ be a graph drawn from $\mu_x$.   Then define $\zeta_x(p)=\mu_x(|C(x)|=\infty\text{ in }G_x(p))$.

The threshold regular condition (Definition 6.1. in \cite{abs2021locality}) says that if $\zeta(p-\epsilon)>0$ for some $\epsilon>0$, then there exists $\delta$ such that $\zeta_x(p)\geq \delta$ for all $x\in S$. 
 Then similar to the derivation of \eqref{Russo-1} we get 
 \begin{equation}\label{Russo-Bridge}
    p \zeta_x(p)=\mathbb E_{G_x(p)} [b_{p,k}(x)],
 \end{equation}
 {where $b_{p,k}(x)$ is the number of $k$-bridges in $G_x(p)$.}
Now, we are ready to prove the lemma.

For the motif-based graph $G$, define $E^{ext}(G)$ as the set of external edges in $G$. 
Let $e \in E^{ext}(G)$ be an external $k$-bridge in $G$. Since $G^{ext}$ is a tree, then there exists a $k'<k$ such that $e$ is a $k'$-bridge for $G^{ext}$. Further, we can find a lower bound on $k'$ by Condition~\ref{cond: internal}.
By this condition, each motif can add at most $S_{\max}$ to the distance of the root $o$ from the boundary. Therefore a $k$-bridge in $G$ is  an external $k'$-bridge in $G^{ext}$, where $k'\geq k/S_{\max}$.  
Then because the number of $k'$-bridges is decreasing in $k'$,
\[\mathbb E_{G_x(p)}[b^{ext}_{p,k}(x)]\leq \mathbb E_{G^{ext}_x(p)}[b_{p,\frac{k}{S_{\max}}}(x)].\]
By applying {Lemma C.2. in \cite{abs2021locality}} there exists $L$ such that for all $p\geq p_c(\mu^{ext})+\epsilon$ and all $x\in S$ and all $k\geq 0$
\[\mathbb E_{G_x(p)}[b^{ext}_{p,k}(x)]\leq \mathbb E_{G^{ext}_x(p)}[b_{p,\frac{k}{S_{\max}}}(x)]\leq L,\]
which bounds the number of external $k$-bridges in $G_x(p)$.

Now, to bound the number of internal $k$-bridges, note that each internal bridge $e$ can be mapped to the an external $k$-bridge which connects the motif containing $e$ to its ancestor motif.  Further by Condition~\ref{cond: internal}  there are at most ${S_{\max}\choose 2}$ internal edges mapped to the same  external bridge. Therefore,
\begin{equation}\label{eq: bridge bnd motif}
    \mathbb E_{G_x(p)}[b^{int}_{p,k}(x)]\leq {S_{\max}\choose 2}\mathbb E_{G_x(p)}[b^{ext}_{p,k}(x)]\leq S_{\max}^2L.
\end{equation}
As a result, the total number of $k$-bridges is at most $(S_{\max}^2+1)L)$. 

To finish the proof, as in the proof of Lemma 6.2. in \cite{abs2021locality},  we can write,
\begin{align*}
    0\leq\tilde \zeta_k(p')-\tilde \zeta_k(p)
    & =\mu\Big(\int_{p}^{p'}\frac 1{p''}\mathbb E_{G(p)}[b_{p'',k}]dp''\Big)
    \leq\frac 1\epsilon \mu\Big(\int_{p}^{p'}\mathbb E_{G(p)}[b_{p'',k}]dp''\Big)\\
    &\leq \frac{S_{\max}^2L}{\epsilon}|p'-p|,
\end{align*}
where the last inequality follows by \eqref{eq: bridge bnd motif}. 
Since the bound on $\tilde \zeta_k(p')-\tilde \zeta_k(p)$ is uniform in $k$, we conclude that $\zeta(p)$ is Lipschitz continuous and hence  continuous {$p\geq p_c(\mu^{ext})+\epsilon$.  Since $\epsilon>0$ was arbitrary, this gives continuity above $p_c(\mu^{ext})$.}
\qed



\end{document}

%% file: household_fig.tex
\definecolor{my_green}{rgb}{0,0.6,1}
\definecolor{my_orange}{rgb}{1,0.6,0}
\definecolor{my_blue}{rgb}{0.4,0.6,0}
\definecolor{my_grey}{rgb}{0.08235294117647059,0.396078431372549,0.7529411764705882}

\begin{tikzpicture}[line cap=round,line join=round,>=triangle 45,x=1cm,y=1cm, scale=.65]
\draw [line width=1.5pt,color=my_green,fill=my_green,fill opacity=0.55] (1.3,-1.18) circle (0.7cm);
\draw [line width=1.5pt,fill=black,fill opacity=0.36] (2.78,1.49) circle (0.7cm);
\draw [line width=1.5pt,color=my_orange,fill=my_orange,fill opacity=0.44] (-0.22,2.03) circle (0.7cm);
\draw [line width=1.5pt,color=my_blue,fill=my_blue,fill opacity=0.63] (-1.34,-0.28) circle (0.7cm);
\draw [line width=1.5pt] (-1.42,-0.03)-- (-1.23,-0.62);
\draw [line width=1.5pt] (0.8149770912171694,-1.16868359930164)-- (1.327343979159074,-0.6403052461115513);
\draw [line width=1.5pt] (1.327343979159074,-0.6403052461115513)-- (1.1031834656844908,-1.60);
\draw [line width=1.5pt] (1.327343979159074,-0.6403052461115513)-- (1.58,-1.63);
\draw [line width=1.5pt] (1.58,-1.63)-- (1.8557223323491632,-1.1046377383089019);
\draw [line width=1.5pt] (2.68,1.09)-- (2.3,1.55);
\draw [line width=1.5pt] (2.3,1.55)-- (3.06,1.85);
\draw [line width=1.5pt] (3.06,1.85)-- (3.22,1.27);
\draw [line width=1.5pt] (3.22,1.27)-- (2.68,1.09);
\draw [line width=1.5pt] (2.68,1.09)-- (3.06,1.85);
\draw [dashed][line width=1.5pt] (-0.2257681499148244,2.033609450335261)-- (-1.42,-0.03);
\draw [dashed][line width=1.5pt] (-1.234,-0.624)-- (0.814,-1.168);
\draw [dashed][line width=1.5pt] (1.327343979159074,-0.6403052461115513)-- (2.3,1.55);
\draw [dashed][line width=1.5pt] (1.8557223323491632,-1.1046377383089019)-- (3.22,1.27);
\draw [dashed][line width=1.5pt] (-0.2257681499148244,2.033609450335261)-- (2.3,1.55);
\begin{scriptsize}
\draw [fill=my_grey] (-1.42,-0.03) circle (2.5pt);
\draw [fill=my_grey] (-1.2344904605504492,-0.6242937808633668) circle (2.5pt);
\draw [fill=my_grey] (2.3,1.55) circle (2.5pt);
\draw [fill=my_grey] (3.06,1.85) circle (2.5pt);
\draw [fill=my_grey] (3.22,1.27) circle (2.5pt);
\draw [fill=my_grey] (2.68,1.09) circle (2.5pt);
\draw [fill=my_grey] (0.8149770912171694,-1.16868359930164) circle (2.5pt);
\draw [fill=my_grey] (1.8557223323491632,-1.1046377383089019) circle (2.5pt);
\draw [fill=my_grey] (1.58,-1.63) circle (2.5pt);
\draw [fill=my_grey] (1.1031834656844908,-1.6009931610026216) circle (2.5pt);
\draw [fill=my_grey] (1.327343979159074,-0.6403052461115513) circle (2.5pt);
\draw [fill=my_grey] (-0.2257681499148244,2.033609450335261) circle (2.5pt);
\end{scriptsize}
\end{tikzpicture}